\documentclass[twocolumn, 10pt, journal]{IEEEtran}
\IEEEoverridecommandlockouts
\usepackage{lineno}
\usepackage[ruled,vlined]{algorithm2e}
\SetKwInput{KwInput}{Input}
\SetKwInput{KwOutput}{Output}
\usepackage{subfigure}

\usepackage{graphicx}
\usepackage{subfigure}
\usepackage{amsmath}
\usepackage{cuted}
\usepackage{lipsum}
\usepackage{cite}
\usepackage{float}
\usepackage{lettrine}
\usepackage{soul}
\usepackage{color}
\usepackage{url}
\usepackage{multirow}
\usepackage{array}
\usepackage{nomencl}
\makenomenclature
\usepackage{mathtools}
\usepackage{verbatim}
\usepackage[margin=0.6in]{geometry}
\usepackage{amsfonts,amssymb}
\usepackage[table]{xcolor}
\usepackage{physics}
\usepackage{etoolbox}
\renewcommand\nomgroup[1]{%
  \item[\bfseries
  \ifstrequal{#1}{A}{Indices/Sets}{%
  \ifstrequal{#1}{B}{Parameters and constants}{%
  \ifstrequal{#1}{C}{Variables}{%
  \ifstrequal{#1}{D}{Acronyms}}}}%
]}

\usepackage{amsthm}
\usepackage{amssymb}
\newtheorem{theorem}{Theorem}

\newcommand{\myhyperref}[2]{\hyperref[#1]{#2 \ref*{#1}}}

 \usepackage[normalem]{ulem}

\begin{document}

\title{\huge Provably Quantum-Secure Microgrids through Enhanced Quantum Distributed Control}

\author{Pouya Babahajiani,~\IEEEmembership{Senior Member,~IEEE,} P. Zhang, Ji Liu,~\IEEEmembership{Member,~IEEE,} and \par Tzu-Chieh Wei  % <-this % stops a space
\thanks{This work was supported in part by the US Department of Energy's Office of Electricity under Agreement  No. 37533, in part by the US National Science Foundation under Grant Nos. ECCS-2018492, OIA-2134840 and OIA-2040599, and in part by Stony Brook University's Office of the Vice President for Research through a Quantum Information Science and Technology Seed Grant. T.-C. Wei was supported by the National Science Foundation under Grant
No. PHY-1915165.
\par
 P. Babahajiani, P. Zhang and J. Liu are with the Department of Electrical and Computer Engineering, Stony Brook University, NY 11794, USA (e-mails: pouya.babahajiani, p.zhang, ji.liu@stonybrook.edu).
 
 T.-C. Wei is with C.N. Yang Institute for Theoretical Physics, Stony Brook University, NY 11794-3800, USA (e-mail: tzu-chieh.wei@stonybrook.edu).

 }}

\maketitle

\begin{abstract}
Distributed control of multi-inverter microgrids has attracted considerable  attention  as  it  can  achieve  the  combined  goals of  flexible  plug-and-play  architecture  guaranteeing  frequency and  voltage  regulation  while  preserving power sharing among nonidentical distributed energy resources (DERs). However, it turns out that cybersecurity has emerged as a serious concern in distributed control schemes. Inspired by quantum communication developments and their security advantages, this paper devises a scalable quantum distributed controller that can guarantee synchronization, and power sharing among DERs. The key innovation lies in the fact that the new quantum distributed scheme allows for exchanging secret information directly through quantum channels among the participating DERs, making microgrids inherently cybersecure. Case studies on two ac and dc microgrids verify the efficacy of the new quantum distributed control strategy.

\end{abstract}
\hfill

\vspace{-1pt}
\IEEEpeerreviewmaketitle
\normalsize

\lettrine{M}{icrogrids} have emerged as a promising new paradigm of electricity resiliency which deliver a growing share of the energy~\cite{morstyn2018using} and facilitate the penetration of the renewable resources into grid utility. In this spirit, in order to match up with their main characteristics including flexibility and scalability, microgrids are becoming ever more reliant on distributed control frameworks~\cite{9352975, 9462515}, and hence they have become cyber-physical systems that requires complicated network technologies to handle massive utilization of  communication  and  computation  devices~\cite{9026756}.

 On the other hand, while distributed control strategies can enhance microgrids resilience, they may cause cybersecurity challenges since they can be vulnerable to cyber attacks on communication links. Malicious  signals from third party agents can drive the microgrid toward inconsistent performance and instability of the whole system \cite{9158562}.  

Finding solutions to encounter cybersecurity issues in microgrids with distributed control strategies is an ongoing research \cite{9462515, chen2020fdi, sahoo2020resilient}. However, the existing solutions may become insecure due to the rapid development of supercomputers and the emergence of quantum computers ~\cite{wright2019benchmarking, qi2019implementation} and so they can make traditional/classical methods obsolete. Utilizing principles of quantum mechanics, quantum communication offers provable security of communication and is a promising solution to counter such threats \cite{qi2019implementation}.

Quantum physics principles give rise to novel capability unachievable with classical transmission media. In these schemes, information is encoded in the particle’s quantum state, which cannot be copied, and any attempt to do so can be detected. Therefore, the critical aspect is unconditional information security which is impossible with classical information processing~\cite{awschalom2021development}. The %most exciting 
key benefit of using quantum-secured information is that the lifetime of the security is “infinite”, i.e., it will be secure against any future advance in computation capability~\cite{awschalom2021development}.

Therefore, quantum communication~\cite{Popkin1026, yu2020entanglement, castelvecchi2021quantum} has proposed a revolutionary step in secure communication due to its high security of the quantum information, and many models including quantum key distribution (QKD)~\cite{chen2021integrated}, quantum teleportation~\cite{luo2019quantum}, discrete-variable quantum secure direct communication~\cite{qi202115, hu2016experimental} and continuous variable quantum secure direct communication~\cite{pirandola2008continuous} have been developed. Based on QKD technology, many different types of quantum communication networks have been proposed \cite{joshi2020trusted, chen2021integrated}. However, these communication networks based on QKD technology only transmit the key, but do not directly transmit information. On the other hand, quantum secure direct communication is a kind of information carrier with quantum state in communication. In this method, secret information is directly transmitted over a secure quantum channel and inspite of QKD schemes, they do not require key distribution and key storage \cite{qi202115}.

Recent development in quantum physics suggests the possibility of investigating consensus problems for networks of quantum nodes~\cite{6849451, mazzarella2015extending, 7109119}. In~\cite{7109119}, it is shown that consensus can be obtained for a network of quantum nodes through establishing interactions among them by means of swapping operators, so that the whole network is described by a Lindblad master equation with the Lindblad terms generated by the swapping operators. Lindblad equation~\cite{BreuerH, lindblad1976} is utilized to describe the state evolution of quantum systems which have interactions with environment. Another approach is exploiting the gossip-type interaction between neighboring quantum nodes~\cite{6849451}. 

The existing approaches, however, are not utilizable for the problem of distributed control in microgrids for two main reasons. First, as we will elaborate more in the following sections, the distributed control problem in microgrids is a distributed tracking problem within which, both possibly time-varying targets and coupling mechanisms exist. Although in~\cite{7109119}, the swapping operator is exploited to build an underlying interaction graph and therefore the coupling mechanism for the quantum network, for the zero Hamiltonian case, the network's final state is the average of the initial states, and under the non-zero Hamiltonian, each qubit tends to the same trajectory related to the network Hamiltonian and initial states of all the nodes.  The latter mostly is difficult to derive the explicit trajectory of each qubit~\cite{7109119}. Second, these schemes are valid as long as the corresponding quantum system is not measured, which is necessary as at each time step, the information needs to be extracted (or encoded) from (or into) the quantum nodes, which makes these schemes impractical for realistic distributed control of microgrids.

Inspired by %these recent developments and 
the existing cybersecurity challenges in microgrids with classical distributed control frameworks, in a recent work~\cite{9850415}, we devised a quantum distributed control for ac and dc microgrids within which, the information carrier is quantum states and the transmission media is a quantum channel, i.e., information is encoded into quantum states which are directly sent over quantum channels among participating DERs, while the control objectives including power/current sharing and frequency/voltage regulation are guaranteed. In this method, there is an assumption that at each time-step, qubits are (re-)initialized on the first quarter of the equator of the Bloch sphere and therefore the evolution of the qubits at each node on the equator is translated into control signal.

In this paper, we demonstrate, however, that with this assumption, the synchronization might not be achieved if quantum states become mixed over the time steps. Furthermore, it might be possible for a third party agent to find out that information is being encoded into the phase angle of the exchanged qubits. We aim to show that by some modifications and relaxing this assumption, 
%{\color{cyan}[highlight in the paper the old and new assumptions]}
the problem of mixed states can be tackled and also the security of the algorithm can be significantly enhanced such that even if a third party agent can measure the exchanged qubits, the measurement outcomes would be some random values which do not reveal information to the eavesdropper. Consequently, adversaries can not access nor falsify the data exchanged among DERs. The key contribution is therefore a provably quantum-secure distributed control that enables unprecedentedly resilient and cybersecure microgrids.

The rest of the paper is organized as follows. 
%Section II provides some preliminaries including relevant concepts in graph theory and quantum systems along with notations and conventions we use in this paper. 
In Section II we formulate the problem by reviewing our recently proposed quantum distributed control in~\cite{9850415} and then pointing out the issues. Section III presents the developed quantum-secure distributed controller (QSDC) together with a numerical example. Section IV is devoted to explain the devised QSDC for ac and dc microgrids. Simulation results are also provided. Section~V concludes the paper.

\section{Problem Formulation}

To formulate our problem, we present an overview of the developed quantum distributed controller proposed in~\cite{9850415} followed by a discussion on the issues that can impair the synchronization of the network under this scheme. Fundamental materials from graph theory \cite{Godsil} and quantum systems \cite{nielsen_chuang_2010} used throughout the paper can be found in Appendix A.

\subsection{Quantum Distributed Control}

Distributed control problems of microgrids can be described as a network of differential equations over a simple, connected graph $G = (V,E)$ whose node set $V=\{ v_1, v_2, \ldots, v_n\}$ represents $n$ DERs and edge set $E$ depicts allowable communication among the DERs. These DERs evolve through interactions, according to certain networking scheme and dynamics.
\begin{figure}
	\centering
	\includegraphics[trim=15 15 20 23,clip, width=0.4\textwidth, height=0.28\textwidth]{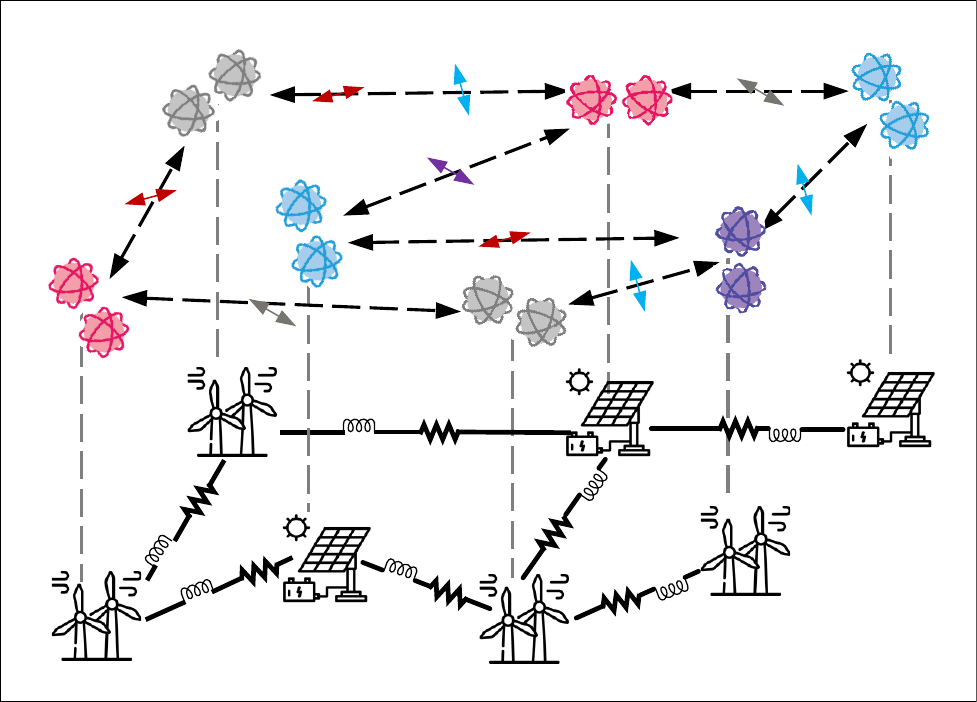}
	\caption{\small Coupling of the physical microgrid to the network of quantum controllers - Information is exchanged through quantum communication.}
	\label{fig4}\vspace{-10px}
\end{figure}

As a representative example, the goal of the distributed frequency control problem in ac microgrids is to regulate the network frequency to a rated value and guarantee active power sharing among the DERs, whose dynamics can be described by
\begin{equation}
\label{eq000}
    \begin{aligned}
    &\omega_i=\omega^*-n_iP_i+\Phi_i, \\
    &\dot\Phi_i=f(\Phi_i,\; P_i, \; \Phi_j,\;j\in N_i),
    \end{aligned}
\end{equation}
where $\omega_i$ represents the frequency at DER$_i$, $\omega^*$ is a nominal network frequency, $P_i$ is the active power injection at DER$_i$, $n_i$ is the droop gain, and the dynamics of $\Phi_i$ represents the secondary control variable, which is a function ($f$) of its current value, $P_i$,  and its neighbors' values $\Phi_j$'s which eventually evolves toward an (weighted) average of its neighbours values $\Phi_j$ such that all control variables $\Phi_i$ converge to the common value $\Phi_i=\Phi_j=n_iP_i$. 

By the same token, the problem of distributed voltage control and power/current sharing in DC microgrids, can be formulated similar to that of Eq.~(\ref{eq000}), as Eq.~(\ref{590000}) which will be discussed more in Section IV:
\begin{equation}
    \label{590000}
    \begin{aligned}
    &V_i^{ref}=V^*-m_iI_i+\Phi_i,\\
    &\dot\Phi_i=f(\Phi_i,\; I_i, \; \Phi_j,\;j\in N_i),
    \end{aligned}
\end{equation}
where $V_i^{ref}$ is the generated voltage reference, $V^*$ is the nominal dc voltage, $m_i$ is the current droop gain, and $I_i$ is the output current of DER$_i$. 
%and $\Phi_i$ again denotes the secondary control variable

In~\cite{9850415}, the goal is to construct a network of differential equations, such as those  in Eqs.~(\ref{eq000}) and (\ref{590000}), to control a network of DERs as shown in Fig.~\ref{fig4}. There, in contrast to the classical synchronization, quantum bits are what is exchanged among the nodes and the framework is formulated as follows.

\noindent
$\bullet$ In this framework, each microgrid is equipped with or connected to a quantum computing (QC) device, which prepares a quantum state and then seeks a consensus among all the QCs in a distributed manner. The state of each quantum device can be described by a positive Hermitian density matrix $\rho$. Since synchronization requires interaction among all quantum devices, it is assumed that each device has access to the (quantum) information of its neighbors. Let $\ket{\psi}=\ket{q_1q_2 \cdots q_n}$ be the state of the whole quantum network and $\rho=\ket{\psi} \bra{\psi}$. The  following Lindblad master equation was introduced to construct  the  network  of  differential  equations (For more information on Markovian master equations in Lindblad form, which is a suitable way to describe the dynamics of a quantum system with interaction with environment, see \cite{wiseman2009quantum}): \vspace{-5pt}
\begin{equation}
\label{eq500}
    \begin{aligned}
    \dot \rho (t)&=\sum_{i=1}^n \Big(C_i^{}\rho C_i^{\dagger} - \frac{1}{2}\{ C_i^{\dagger}C_i^{}, \rho \} \Big)\\
    &+ \sum_{\{ i,j \}\in E} \Big(C_{i,j}^{}\rho C_{i,j}^{\dagger} - \frac{1}{2}\{ C_{i,j}^{\dagger}C_{i,j}^{}, \rho \} \Big).
    \end{aligned}
\end{equation}
where $C_i$ and $C_{i,j}$ are unitary jump operators.

In this scheme, the state of each quantum node at each time step is updated as follows: 
\begin{equation}
\label{eq112}
\begin{aligned}
    \ket{q_i(t)}=\cos{\frac{\pi}{4}}\ket{0} + e^{\imath\phi_i(t)}\sin{\frac{\pi}{4}}\ket{1} , \;\; t\in\{0,1,2,\ldots\}, %\textcolor{olive}{!!},
    \end{aligned}
\end{equation}
which is the general state in polar coordinates set on the xy-plane, where $\phi_i(0)\in(0,\pi/2)$ and each $\phi_i(t)$, $t\ge 1$, is inferred from the averaged measurement outcome at node $i$.

In Eq.~(\ref{eq500}), $C_{i,j}$ is the swap operator that specifies the external interaction between quantum computing devices $i$ and $j$ such that
\begin{equation}
    \label{eq222}
    \begin{aligned}
    C_{i,j}(\ket{q_1} &\otimes ... \otimes \ket{q_i}\otimes ... \otimes\ket{q_j} \otimes ... \otimes\ket{q_n})\\
    & =\ket{q_1} \otimes ... \otimes \ket{q_j}\otimes ... \otimes\ket{q_i} \otimes ... \otimes\ket{q_n},
    \end{aligned}
\end{equation}
%and $E$ is the set of communication links within the quantum network. 
and the jump operator, $C_{i}$, is defined as
\begin{equation}
\label{eq7}
\begin{aligned}
    C_{i}=I^{\otimes (i-1)} \otimes R_z(\phi) \otimes I^{\otimes (n-i)},
    \end{aligned}
\end{equation}
with $R_z(\phi)$ being the single-qubit rotation-Z operator by an angle $\phi$ radians around the Z-axis. By definition, the operator $C_i$ acts only on $\ket{q_i}$ without changing the states of other qubits. As can be seen, the jump operators $C_i$ are state dependent and are updated based on the target values $\phi_{i,t}$ and the measured $\phi_i(t)$. Furthermore, as mentioned earlier, at the beginning of each time step, all the qubits are re-initialized as in (\ref{eq112}) based on the measurement outcome of the previous step. Therefore, at each time step, the master equation components are updated based on the target values and the obtained measurement signals. Thus, the density matrix at time $t+dt$ can be decomposed into $\rho(t+dt)=\rho(t)+d\rho_t$, where $d\rho(t)$ is defined in (\ref{eq500}).

Next, to obtain %the angles 
$\phi_i$, the following observables are utilized:  
\begin{align}
    &A_{1,i}= I^{\otimes (i-1)} \otimes \sigma_x  \otimes I^{\otimes (n-i)},\\ 
    &A_{2,i}= I^{\otimes (i-1)} \otimes \sigma_y  \otimes I^{\otimes (n-i)} \label{26b}.
\end{align}
The operator $I^{\otimes (i-1)} \otimes \sigma_{x/y}  \otimes I^{\otimes (n-i)}$ acts only on $\ket{q_i}$ where node-wise means, having $\sigma_x$ and $\sigma_y$ which are Pauli matrices as observers at each node, \vspace{-3pt}
\begin{equation}
\label{eq15}
    \begin{aligned}
    \sigma_x=\begin{pmatrix}
    0 & 1\\
    1 & 0\end{pmatrix}, \;\;\;\;\; \sigma_y=\begin{pmatrix}0 & -\imath\\
    \imath & 0\end{pmatrix}.
    \end{aligned}
\end{equation}

The expectation value of an observable $A$ in a state, represented by a density matrix $\rho$, is given by $\langle A \rangle = \rm tr(\rho A)$~\cite{preskill}. For a general one qubit state $\rho_i$, ${\rm tr}(\rho_i \sigma_x) = r_i \sin\theta_i \cos\phi_i$, ${\rm tr}(\rho_i \sigma_y) = r_i \sin\theta_i \sin\phi_i$ and ${\rm tr}(\rho_i \sigma_z) = r_i \cos\theta_i$ where $r_i$ is the size of the state vector $i$. Generally, Lindblad equation results in states becoming more mixed; however, we only let the system evolve in a short time and re-initialize the system in a product of pure qubit states. Therefore, we can consider $r=1$ and $\theta=\pi/2$ and hence \vspace{-3pt}
\begin{equation}
\label{260aa}
    \begin{aligned}
    &\rm tr(\rho\sigma_x) =\cos{\phi_i}, \;\;\;\;\;\rm tr(\rho\sigma_y) =\sin{\phi_i}.
    \end{aligned}
\end{equation}
%Furthermore, (\ref{260a}) and (\ref{260b}) 
which are equivalent to $\rm tr(\rho A_{1,i})=\cos{\phi_i}$ and  $\rm tr(\rho A_{2,i})=\sin{\phi_i}$, respectively. Note that $\frac{d}{dt} \langle A \rangle = \frac{d}{dt} \rm tr(\rho A) = \rm tr(\dot \rho A)$. Hence, it can be readily obtained that, $\frac{d}{dt}\arctan (\frac{\rm tr(\rho A_{2,i})}{\rm tr(\rho A_{1,i})})$ gives the dynamic of $\phi_i$, or the synchronization rule, as follows:\vspace{-5pt}
 \begin{equation}
\label{eq35000}
    \begin{aligned}
    \dot \phi_i=\sin{(\phi_{t,i}-\phi_i)}+\sum_{j=1}^n a_{i,j}\sin{(\phi_j-\phi_i)}. 
    \end{aligned}
\end{equation}  
where $a_{i,j}=1$ if $C_{ij}\neq \mathbf{0}$ and $a_{i,j}=0$ otherwise.
\hfill
$\blacksquare$

It is worth noting that both $\rm tr(\rho A_{1,i})$ and $\rm tr(\rho A_{2,i})$ are used in (\ref{eq35}) to find the trajectory that $\phi_i$ traverses along the time; however, either $\arccos{(\rm tr(\rho\sigma_x))}$ or $\arcsin{(\rm tr(\rho\sigma_y))}$ gives $\phi_i$.

As mentioned, at each time-step, qubits are (re-)initialized on the equator of the Bloch sphere and hence, remain on the equator along the process. Fig.~(\ref{fig33}) shows a numerical example of two qubits (utilizing %the Python-based 
open source software QuTiP~\cite{qutip}) with the initial states $\ket{q_1}=\frac{1}{\sqrt{2}}[1, \; 1]^T$ (node 1) and $\ket{q_2}=\frac{1}{\sqrt{2}}[1, \; \imath]^T$ (node 2) that get synchronized to the target $\phi^*=\pi/6$ while they remain on the equator along the time.

\begin{figure}
	\centering
	\includegraphics[trim=10 3 10 20,clip, width=0.5\textwidth, height=0.245\textwidth]{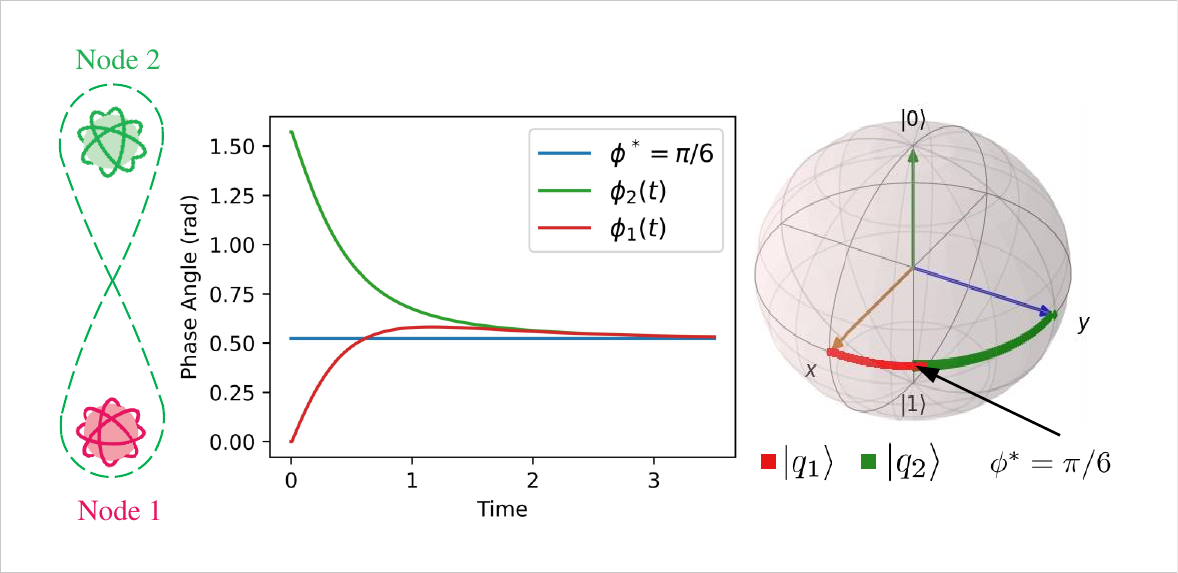} \vspace{-20pt}
	\caption{\small Exponential synchronization of the phase angles to the target phase $\phi^*=\pi/6$ - (Left side) Connected communication formed by the swap operator - (Right side) Bloch sphere representation shows the trajectories traversed by the qubits along the time.}
	\label{fig33}\vspace{-10px}
\end{figure}

\vspace{-10pt}

\subsection{The Problem}

There are two issues with the framework just described. First, in case qubits become mixed, neither $\arccos{(\rm tr(\rho\sigma_x))}$ nor $\arcsin{(\rm tr(\rho\sigma_y))}$ gives $\phi_i$ and synchronization would not be attained even if $\theta_i$ remains $\pi/2$ since, $\arccos{(\rm tr(\rho\sigma_x))}=\arccos{(r_i\cos{\phi_i})}\neq\phi_i$ and Eq.~(\ref{eq35000}) is no longer valid. For instance, Fig. (\ref{fig033}) shows how for the two qubits example in Fig. (\ref{fig33}), synchronization would be violated if $\ket{q_2}$ slightly becomes mixed at some random steps after the master equation evolution.
\begin{figure}
	\centering
	\includegraphics[trim=70 450 150 55,clip, width=0.4\textwidth, height=0.26\textwidth]{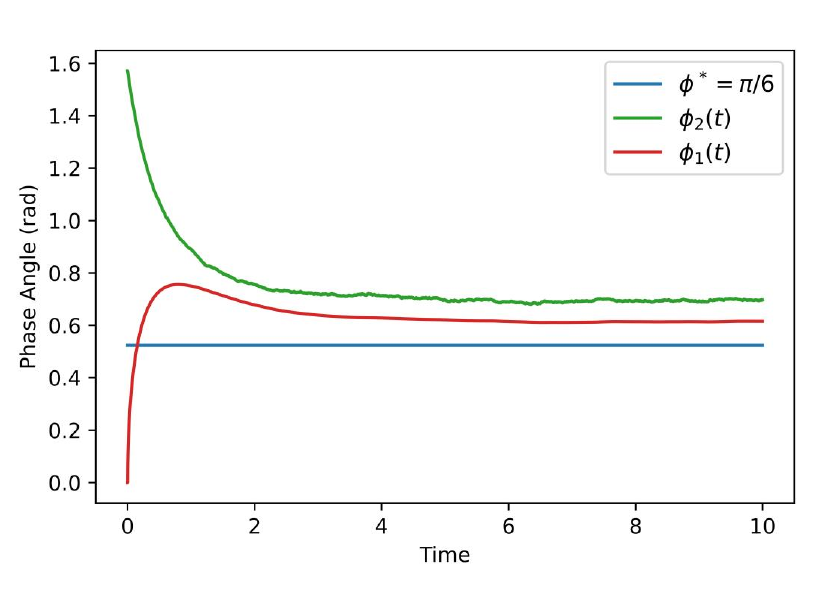} \vspace{-10pt}
	\caption{\small Violation of synchronization when $\ket{q_2}$ becomes mixed at some steps after the master equation evolution.}
	\label{fig033}\vspace{-10px}
\end{figure}
Second, if the eavesdropper does not know the basis, she can figure out the angle $\phi_i$ if she has a lot number of qubits and performs state tomography~\cite{torlai2018neural} to determine the basis and the angle $\phi_i$, allows her to perform false data injection and hence destabilize the microgrid.

In the next section, we demonstrate that by some modifications and relaxing the assumption of (re-)initializing qubits on the equator of Bloch sphere, the security of the algorithm can be significantly enhanced such that even if a third party agent can measure the exchanged qubits, the measurement outcomes would be some random values which do not reveal information to the eavesdropper. Furthermore, even if the states become mixed after the master equation evolution, or the (re-)initialization produces mixed states, the synchronization rule remains valid.

\vspace{-5pt}

\section{Quantum-Secure Distributed Control}

In our developed scheme, the state of each quantum node at each time step is prepared as follows
\begin{equation}
\label{eq1112}
\begin{aligned}
    \ket{q_i(t)}=\cos{\frac{\theta_i(t)}{2}}\ket{0} + e^{\imath\phi_i(t)}\sin{\frac{\theta_i(t)}{2}}\ket{1} , \;\; t\in\{0,1,2,\ldots\},
    \end{aligned}
\end{equation}
which is the general state in polar coordinates set on the surface of the Bloch sphere, where $0<\theta_i(t)<\pi$ is a random value and $\phi_i(0)\in(0,\pi/2)$.

\begin{theorem}
Consider the master equation (\ref{eq500}). Let each $\phi_i(t)$, $t\ge 1$, be the averaged measurement outcome  at node $i$, utilizing the observables (7) and (\ref{26b}), such that $\phi_i = \arctan (\frac{\rm tr(\rho A_{2,i})}{\rm tr(\rho A_{1,i})})$. Then, we have the following:
\begin{enumerate}
    \item The dynamic of $\phi_i$ can be written as \vspace{-10pt}
    \begin{equation}
    \label{260bbb}
    \begin{aligned}
    \dot \phi_i=\sin{(\phi_{t,i}-\phi_i)}+\sum_{j=1}^n a_{i,j}\frac{r_j\sin{\theta_j}}{r_i\sin{\theta_i}}\sin{(\phi_j-\phi_i)};
    \end{aligned}
    \end{equation}
    
    \item The synchronization rule in \textup{1)} remains valid even if the (re-)initialization produces mixed states, or states become mixed after master equation evolution;
    
    \item The random angle $\theta_i$ randomizes the state and hence on average (over $\theta_i$) the state appears to be random to a third party who measures the exchanged qubits among the quantum nodes.
\end{enumerate}

\end{theorem}

\begin{proof}
As mentioned, for the observables $A_{1,i}$ and $A_{2,i}$, we have the following expectation values: \vspace{-6pt}
\begin{equation}
\label{260a}
    \begin{aligned}
    &\rm tr(\rho A_{1,i})=r_i\sin{\theta_i}\cos{\phi_i}, \\
    &\rm tr(\rho A_{2,i})=r_i\sin{\theta_i}\sin{\phi_i}.
    \end{aligned}
\end{equation}
Considering (\ref{eq500}) and that $\frac{d}{dt} \langle A \rangle = \rm tr(\dot \rho A)$, we can obtain the synchronization rule as follows:
\begin{equation}
\label{eq30}
    \begin{aligned}
    \rm tr(\dot \rho A_{1,i}) &=r_i\cos{\phi_{t,i}}\sin{\theta_i}-r_i\cos{\phi_i}\sin{\theta_i}+\\
    &\sum_{j=1}^n a_{i,j}(r_j\cos{\phi_j}\sin{\theta_j}-r_i\cos{\phi_i}\sin{\theta_i}),\\
    \rm tr(\dot \rho A_{2,i}) &=r_i\sin{\phi_{t,i}}\sin{\theta_i}-r_i\sin{\phi_i}\sin{\theta_i}+\\
    &\sum_{j=1}^n a_{i,j} (r_j\sin{\phi_j}\sin{\theta_j}-r_i\sin{\phi_i}\sin{\theta_i}),
    \end{aligned}
\end{equation}
where $a_{i,j}=1$ if $C_{ij}\neq \mathbf{0}$ and $a_{i,j}=0$ otherwise. Utilizing (\ref{260a}) and (\ref{eq30}), the dynamic of $\phi_i$ can be found as follows: \vspace{-10pt}

\begin{equation}
\label{eq35}
    \begin{aligned}
    \dot \phi_i &= \frac{d}{dt}\arctan (\frac{\rm tr(\rho A_{2,i})}{\rm tr(\rho A_{1,i})}) \\
    &= \frac{\rm tr(\dot \rho A_{2,i})tr(\rho A_{1,i})-tr(\dot \rho A_{1,i})tr(\rho A_{2,i})}{r_i^2\sin^2{\theta_i}\sin^2{\phi_i}+r_i^2\sin^2{\theta_i}\cos^2{\phi_i}}\\
    &=\frac{1}{r_i^2\sin^2{\theta_i}}\Bigg[\bigg( r_i\sin{\phi_{t,i}}\sin{\theta_i}-r_i\sin{\phi_i}\sin{\theta_i}\\
    &\;\;+\sum_{j=1}^n a_{i,j} (r_j\sin{\phi_j}\sin{\theta_j}-r_i\sin{\phi_i}\sin{\theta_i})\bigg)r_i\cos{\phi_i}\sin{\theta_i} \\
    &\;\;\;\;\;\;\;-\bigg(r_i\cos{\phi_{t,i}}\sin{\theta_i}-r_i\cos{\phi_i}\sin{\theta_i}+\\
    &\;\;\;\;\sum_{j=1}^n a_{i,j} (r_j\cos{\phi_j}\sin{\theta_j}-r_i\cos{\phi_i}\sin{\theta_i})\bigg)r_i\sin{\phi_i}\sin{\theta_i}\Bigg] \\
    &=\sin{\phi_{t,i}}\cos{\phi_i}-\cos{\phi_{t,i}}\sin{\phi_i}+\\
    &\;\;\;\;\;\;\;\;\;\;\; \sum_{j=1}^{n}a_{i,j}\frac{r_j\sin{\theta_j}}{r_i\sin{\theta_i}}(\sin{\phi_j\cos{\phi_i}}-\cos{\phi_j}\sin{\phi_i}) \\
    &=\sin{(\phi_{t,i}-\phi_i)}+\sum_{j=1}^n a_{i,j}\frac{r_j\sin{\theta_j}}{r_i\sin{\theta_i}}\sin{(\phi_j-\phi_i)}.
    \end{aligned}
\end{equation}

As stated, $0<\theta_i<\pi$ and hence, $(r_j\sin{\theta_j})/(r_i\sin{\theta_i})>0$ and so  $a_{i,j}\frac{r_j\sin{\theta_j}}{r_i\sin{\theta_i}}\geq0$. Therefore, (\ref{eq35}) can be rewritten as
\begin{equation}
\label{260b}
    \begin{aligned}
    \dot \phi_i=\sin{(\phi_{t,i}-\phi_i)}+\sum_{j=1}^n a_{i,j}^{'}\sin{(\phi_j-\phi_i),}
    \end{aligned}
\end{equation}
where $a^{'}_{i,j}=a_{i,j}\frac{r_j\sin{\theta_j}}{r_i\sin{\theta_i}}>0$ if $C_{ij}\neq \mathbf{0}$ and $a^{'}_{i,j}=0$ otherwise. In Appendix B, it is shown that how the pinning term $\sin{(\phi_{t,i}-\phi_i)}$ forces the phase $\phi_i$ to stick at the value $\phi_{t,i}$ and the coupling mechanism $\sum_{j=1}^n a_{i,j}^{'} \sin{(\phi_j-\phi_i)}$ helps to synchronize the entire system and all the nodes will synchronize to the pinner $\phi_{t,i}$ exponentially fast. 

The above analysis leads to the following critical property of the proposed quantum-secure distributed control:

\begin{enumerate}
    \item With this scheme, the additional random angle $\theta_i$ randomizes the state and hence on average (over $\theta_i$) the state appears to be random, i.e., maximally mixed $I/2$. Then, the eavesdropper cannot figure out the correct 0/1 basis and consequently, the information is being encoded into $\phi_i$.
    \item (\ref{eq35}) shows that the scheme works even if the (re-)~initialization produces mixed states, or states become mixed along the master equation evolution.
\end{enumerate}
This completes the proof. 
\end{proof}

Considering (\ref{260a}), even if $\phi_i$ does not change, changing $\theta_i$ impacts the expectation values $\rm tr(\rho A_{1,i})$ and $\rm tr(\rho A_{2,i})$. Therefore, as shown in (\ref{eq35}), both $\rm tr(\rho A_{1,i})$ and $\rm tr(\rho A_{2,i})$ are required to find the trajectory traversed by $\phi_i$ along the time. However, measuring a qubit results in demolishing that quantum state. Hence, to obtain the both expectations at each node, at each time step, two identical qubits will be prepared according to Eq.~(\ref{eq1112}) which will experience the same Master equation evolution. Afterwards, expectation value of $\sigma_x$ is measured for one of the qubits and expectation of $\sigma_y$ for its identical twin. 
 \begin{figure*}
	\centering
	\includegraphics[trim=3 3 3 3,clip, width=1\textwidth, height=0.49\textwidth]{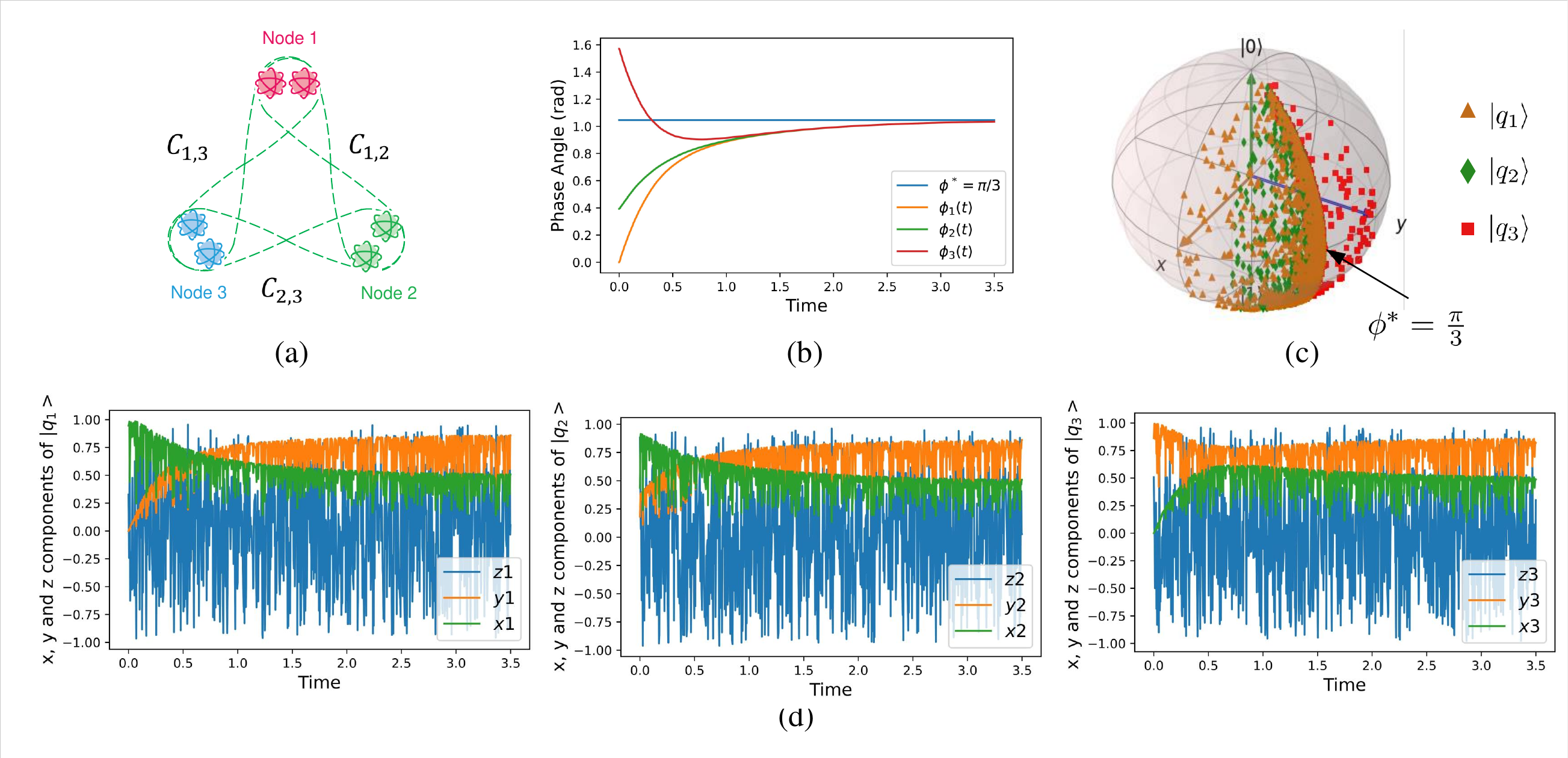}\vspace{-10pt}
	\caption{\small (a) Connected network of three quantum nodes. (b) Exponential synchronization of phase angles to $\phi^*=\pi/3$. (c) Bloch sphere representation of the states of the quantum nodes along the time. It shows how the additional random $\theta_i$ randomizes the states. (d) Measurement outcomes for x,y and z components are random values.}
	\label{fig3}%\vspace{-10px}
\end{figure*}
\begin{algorithm}
%\small
\SetAlgoLined
 %\KwData{thi}
 %\KwResult{gh}
 \textbf{1} At each node, initialize a pair of identical qubits  as a point on the surface of the first quarter of the Bloch Sphere, i.e., $0<\phi_{i}(0)<\pi/2$ and $0<\theta_i(t)<\pi$, using Eq.~(\ref{eq112}) for each adjacent node  \\
 \textbf{2} Transmit quantum information throughout the network such that each quantum node receives a pair of identical qubits from each one of its adjacent nodes \\
 \textbf{3} At each node, update the rotation-Z ($R_z$) operator's argument based on the pinner ($\phi_{t,i}$) and the current value of the phase angle $\phi_i$\\
 \textbf{4} Evolve the master Eq.~(\ref{eq500}) for one time step $\delta t$ by means of the swapping and rotation-Z operators\\
 \textbf{5} Measure the expectation value of the $\sigma_x$ operator for one of the qubits and $\sigma_y$ for its identical twin. Repeating this multiple times and averaging gives the $r_i\sin{\theta_i}\cos{\phi_i}$ and $r_i\sin{\theta_i}\sin{\phi_i}$, respectively. On classical hardware at each node, compute $\arctan(\frac{\langle \sigma_y \rangle}{\langle \sigma_x \rangle})$ to obtain the phase angle $\phi_i$\\
 \textbf{6} Re-initialize the state of each quantum node according to Eq. (\ref{eq112}) and based on the measurement outcomes at step \textbf{5}\\
 \textbf{7} Go back to step \textbf{2}\\
 \caption{Quantum-Secured Distributed Control}
 \label{alg}
\end{algorithm}
The basic outline of the algorithm is summarized in \textbf{Algorithm~1}. 

{\bf Numerical Example.}
In this example, in addition to synchronization, we illustrate that, how choosing random variables for $\theta_i$ results in an unprecedentedly secured distributed control framework. We consider a network composed of three quantum nodes with the following initial states:
\begin{align*}
    &\ket{q_1}=\cos{\frac{1.96}{2}}\ket{0} + \sin{\frac{1.96}{2}}\ket{1},\\
    &\ket{q_2}=\cos{\frac{1.49}{2}}\ket{0} + e^{\frac{\pi}{8}\imath}\sin{\frac{1.49}{2}}\ket{1},\\
    &\ket{q_3}=\cos{\frac{2.07}{2}}\ket{0} + e^{\frac{\pi}{2}\imath}\sin{\frac{2.07}{2}}\ket{1}.
\end{align*}
Here, the target is $\phi^*=\frac{\pi}{3}$. The three qubits interact through swapping operators $C_{1,2}$, $C_{2,3}$ and $C_{1,3}$, forming a connected interaction graph as shown in Fig. \ref{fig3}a. We consider the state of the whole quantum network as $\rho=\ket{q_1q_2q_3}\bra{q_1q_2q_3}$.
As an illustrative example, the swapping operator $C_{1,2}$ would be as follows:
 \begin{equation}
\label{swap}
    \begin{aligned}
    C_{1,2}=\begin{pmatrix} 1 & 0& 0& 0\\ 0 & 0& 1& 0\\0 & 1& 0& 0\\0 & 0& 0& 1 \end{pmatrix} \otimes \begin{pmatrix} 1 & 0\\ 0 & 1 \end{pmatrix}
    \end{aligned}
\end{equation}
 such that,
  \begin{equation}
\label{swap2}
    \begin{aligned}
    C_{1,2} &\ket{q_{\textbf{1}}(t)q_{\textbf{2}}(t)q_3(t)}\bra{q_{\textbf{1}}(t)q_{\textbf{2}}(t)q_3(t)} C_{1,2}^{\dagger}\\
    &-\ket{q_{\textbf{2}}(t)q_{\textbf{1}}(t)q_3(t)}\bra{q_{\textbf{2}}(t)q_{\textbf{1}}(t)q_3(t)}=[\mathbf{0}]_{8\times8}
    \end{aligned}
\end{equation}

 The trajectories of $\phi_1$, $\phi_2$ and $\phi_3$, i.e. $\phi$ angles of $\ket{q_1}$, $\ket{q_2}$ and $\ket{q_3}$ respectively, are sketched in Fig. \ref{fig3}b utilizing the Python-based open source software QuTiP~\cite{qutip}. As illustrated, all the three phase angles converge to $\phi^*$. Therefore, the final state of each quantum node is the state $[\cos{\frac{\theta_i(t)}{2}}, \; e^{\frac{\pi}{3}\imath}\sin{\frac{\theta_i(t)}{2}}]^T$ where, $0<\theta_i(t)<\pi$ is a random value.
 
 As depicted in Fig.~\ref{fig3}d, even if the eavesdropper is able to measure the x, y and z components of the exchanged qubits, the measurement outcomes would be some random values which don't reveal meaningful information to the eavesdropper.
 
To better illustrate that how randomizing $\theta$ increases the security of the scheme, a single qubit measurement experiment is provided in the following (Fig. \ref{figM1}). In this numerical example first, it is shown that how measuring the probabilities of the 0/1 basis of the qubit can give the $\phi$ angle, when qubit is initialized on the equator of the Bloch sphere and there are enough copies of it. Suppose the qubit is initialized as follows:
\begin{align*}
    &\ket{q}=\cos{\frac{\pi}{4}}\ket{0} + e^{\frac{\pi}{6}\imath}\sin{\frac{\pi}{4}}\ket{1}.
\end{align*}
It can be readily obtained that the X, Y and Z components of the qubit can be found through measuring the Z basis of the circuits 1, 2 and 3, respectively, where Z, H (Hadamard gate) and $S^{\dagger}$ are defined as follows:
\begin{equation*}
    \label{measure111}
    \begin{aligned}
    Z=\begin{pmatrix}
    1 & 0\\
    0 & -1\end{pmatrix}, \; H=\frac{1}{\sqrt{2}}\begin{pmatrix}
    1 & 1\\
    1 & -1\end{pmatrix}, \; S^{\dagger}=\begin{pmatrix}
    1 & 0\\
    0 & -\imath\end{pmatrix}.
    \end{aligned}
\end{equation*}
\begin{figure}
\centering
\includegraphics[trim=5 5 5 5,clip,width=0.51\textwidth, height=0.19\textwidth]{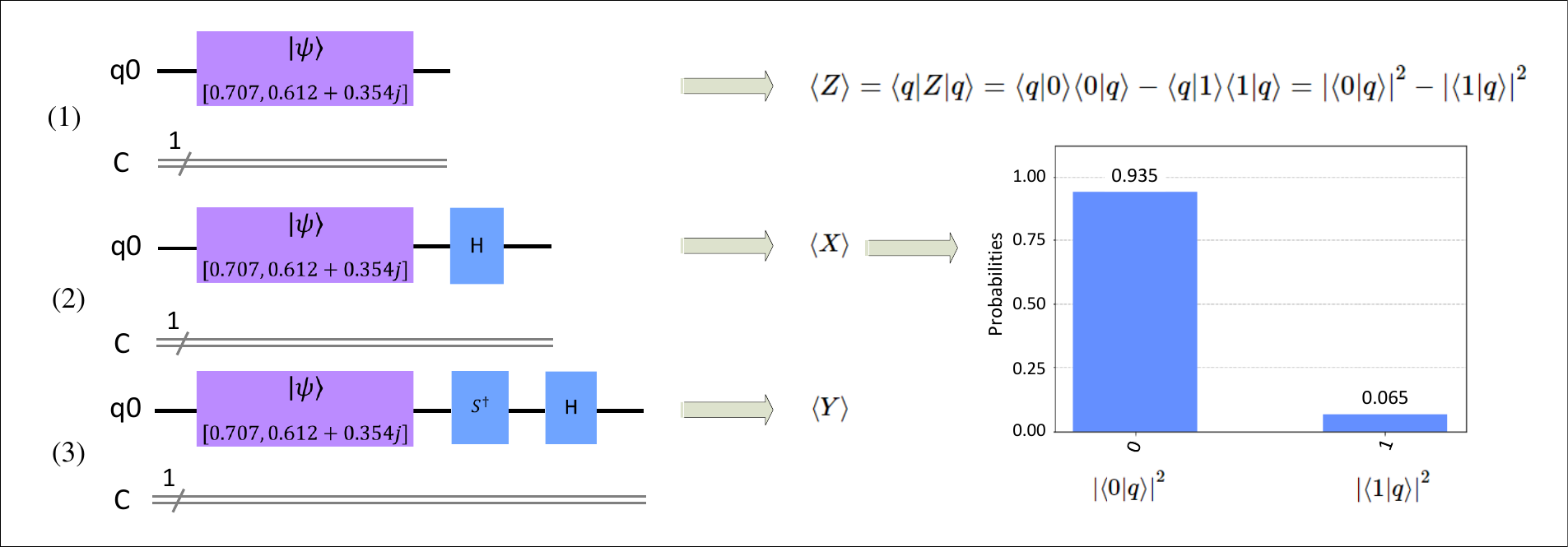}
\caption{\small Measuring the Z basis of the above circuits gives the X, Y and Z components of an arbitrary qubit. The histogram on the right shows the probability of finding the qubit in 0/1 basis with the circuit (2) which is the average over 2000 experiments. Subtracting these two probabilities gives the X component.}
\label{figM1}
\end{figure}
\begin{figure}
\centering
\includegraphics[trim=5 5 5 5,clip,width=0.28\textwidth, height=0.19\textwidth]{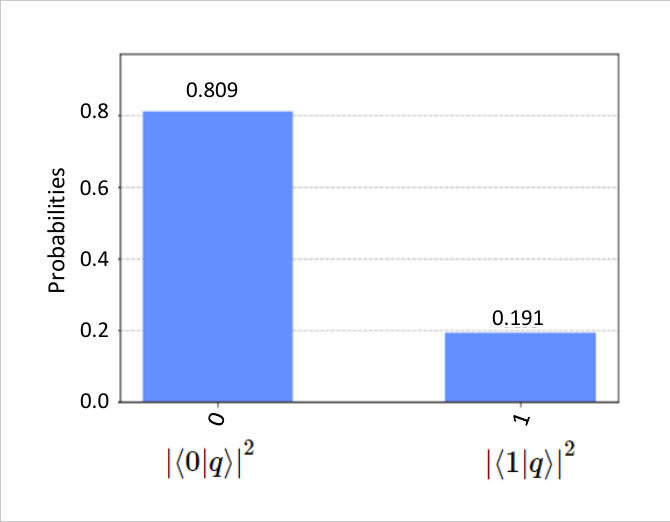}\vspace{-10pt}
\caption{\small Measurement result of circuit (2) with randomized $\theta$.}
\label{figM222}
\end{figure}
As can be seen, the $\phi$ angle of the qubit can be found through measuring the probabilities of finding the qubit in 0/1 basis:
\begin{equation*}
    \label{measure000}
    \begin{aligned}
    \cos^{-1}{(\abs{\bra{0}\ket{q}}^2-\abs{\bra{1}\ket{q}}^2)}=\cos^{-1}{(0.935-0.065)=\frac{\pi}{6}}.
    \end{aligned}
\end{equation*}
Hence, if an eavesdropper has access to enough copies of the quantum state, she can figure out %the 
$\phi$ %angle 
and therefore, the exchanged secret information.
Now, we show how by %utilizing 
using %the 
superposition %feature 
of qubits, we can significantly enhance the security of the scheme. As mentioned, for a general single qubit, the expectation value of $\sigma_x$ (which can be obtained by measuring the Z basis of the circuit (2) and subtracting the probabilities) equals to $\langle \sigma_x \rangle=r\sin{\theta\cos{\phi}}$, and since the qubit is initialized as a pure state on the equator of the Bloch sphere, $r=1$ and $\theta=\pi/2$. But when the angle $\theta$ is randomized such that:
\begin{align*}
    \ket{q}=\cos{\frac{\theta}{4}}\ket{0} + e^{\frac{\pi}{6}\imath}\sin{\frac{\theta}{4}}\ket{1}, \theta \in (0, \pi)
\end{align*}
\begin{figure*}
\centering
\includegraphics[trim=5 5 5 5,clip,width=1\textwidth, height=0.37\textwidth]{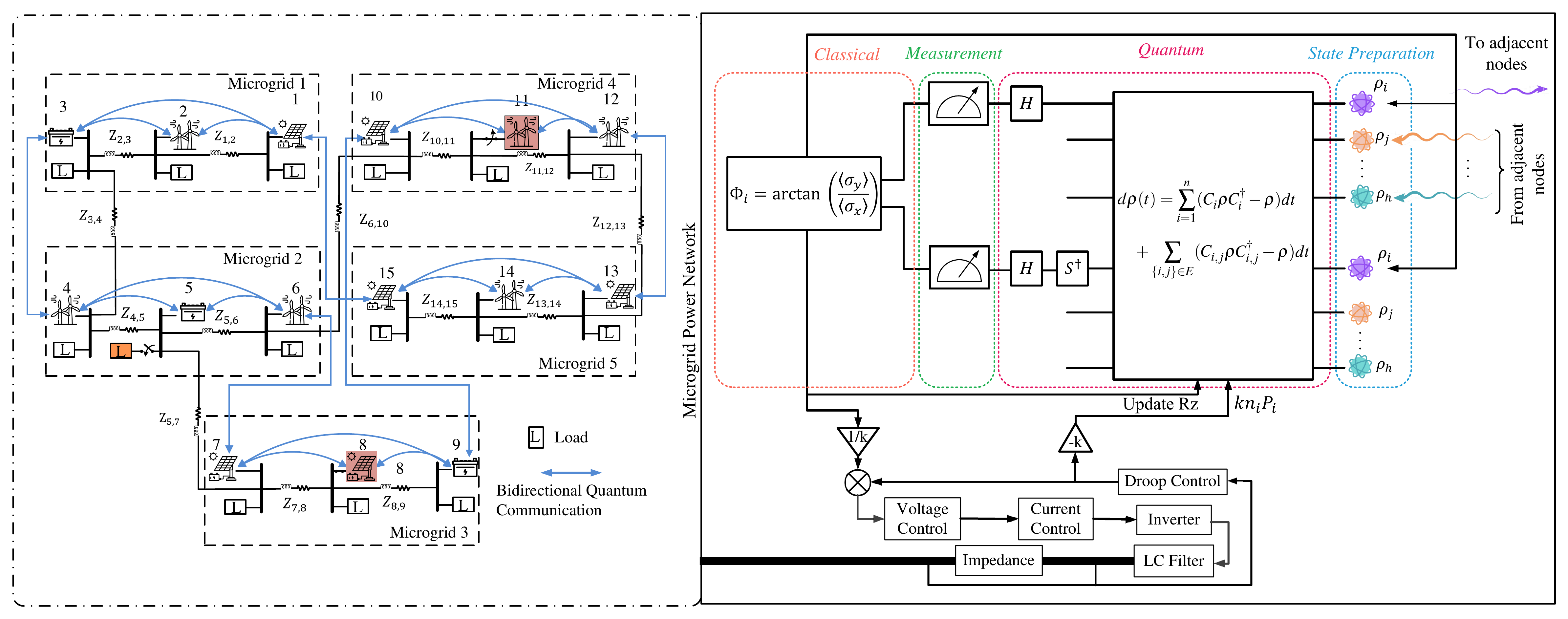}
\caption{\small The AC Networked microgrids case study is shown on the left. Blue bidirectional arrows represent the undirected quantum communications - The figure on the right demonstrates how QSDC is implemented as the secondary controller.}
\label{figM2}
\end{figure*}
$\langle\sigma_x\rangle\neq\cos{\phi}$ and hence, the probabilities of finding the qubit in 0/1 basis do not give information about the X component, and consequently the encoded information. Fig. (\ref{figM222}) shows the probabilities with the circuit (2) which is the average over 2000 shots when random $\theta$ is used at each shot. As it is evident, $\cos^{-1}{(0.809-0.191)}\neq \frac{\pi}{6}$, and from Eve's point of view, exchanged qubits are alternating between $\ket{0}$ and $\ket{1}$ basis with no meaningful pattern, which is also demonstrated in Fig. \ref{fig3}d.

\section{Quantum-Secured Distributed Controller for AC and DC Microgrids}

\subsection{AC Microgrids}

Droop control is a local DER control which is used to regulate the DER's frequency $\omega_i$ through local measurement of the active power injection at DER$_i$ as follows
\begin{equation}
    \label{eq390}
    \begin{aligned}
    &\omega_i=\omega^*-n_iP_i
    \end{aligned}
\end{equation} 
where $\omega^*$ is a nominal network frequency, $P_i$ is the measured active power injection at DER$_i$ and $n_i$ is the gain of the droop coefficient. The secondary control is then utilized to restore the operating frequency to the rated value $\omega^*$. Our developed quantum distributed controller at the secondary level is formulated as follows: 
\vspace{-3pt}
\begin{equation}
    \label{400a}
    \begin{aligned}
    &\omega_i=\omega^*-n_iP_i + \frac{\phi_i}{k}, \\
    &\dot \phi_i=\sin{(kn_iP_i-\phi_i)}+\sum_{j=1}^n a^{'}_{i,j}\sin{(\phi_j-\phi_i)},
    \end{aligned}
\end{equation}
\vspace{-3pt}
\begin{figure*}[h]
\centering \subfigure[]
{\includegraphics[trim=0 20 0 2, clip,width=0.49\textwidth,height=0.16\textwidth]{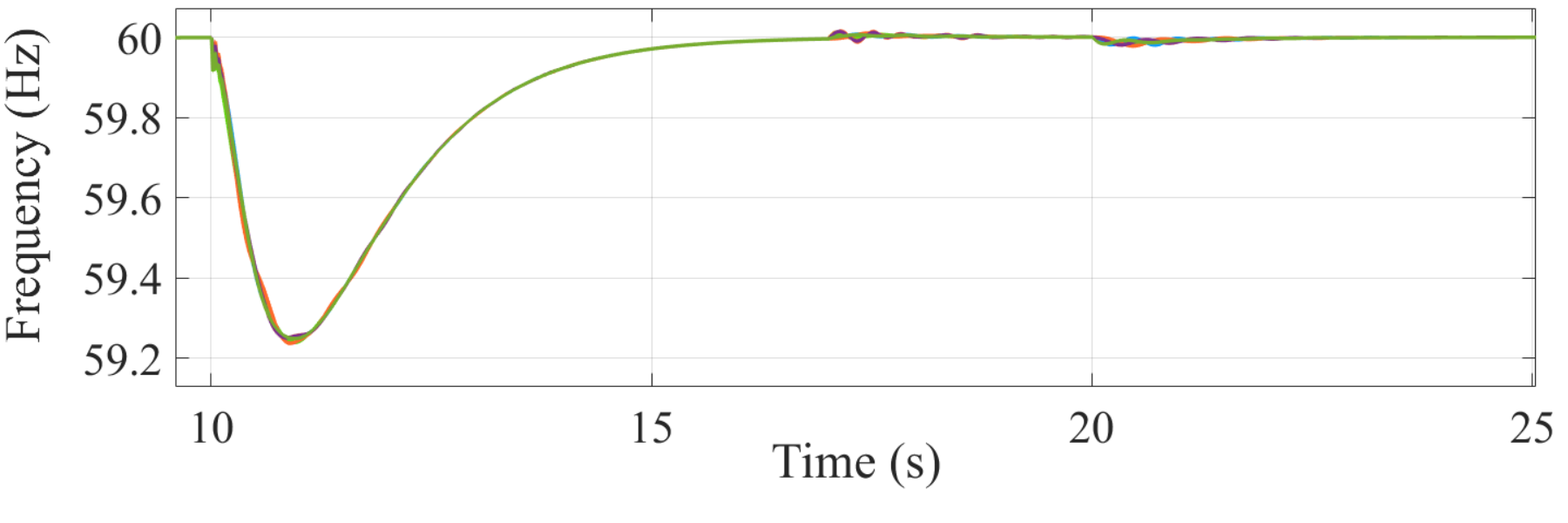}}
\centering \subfigure[]
{\includegraphics[trim=20 125 0 155, clip,width=0.49\textwidth,height=0.16\textwidth]{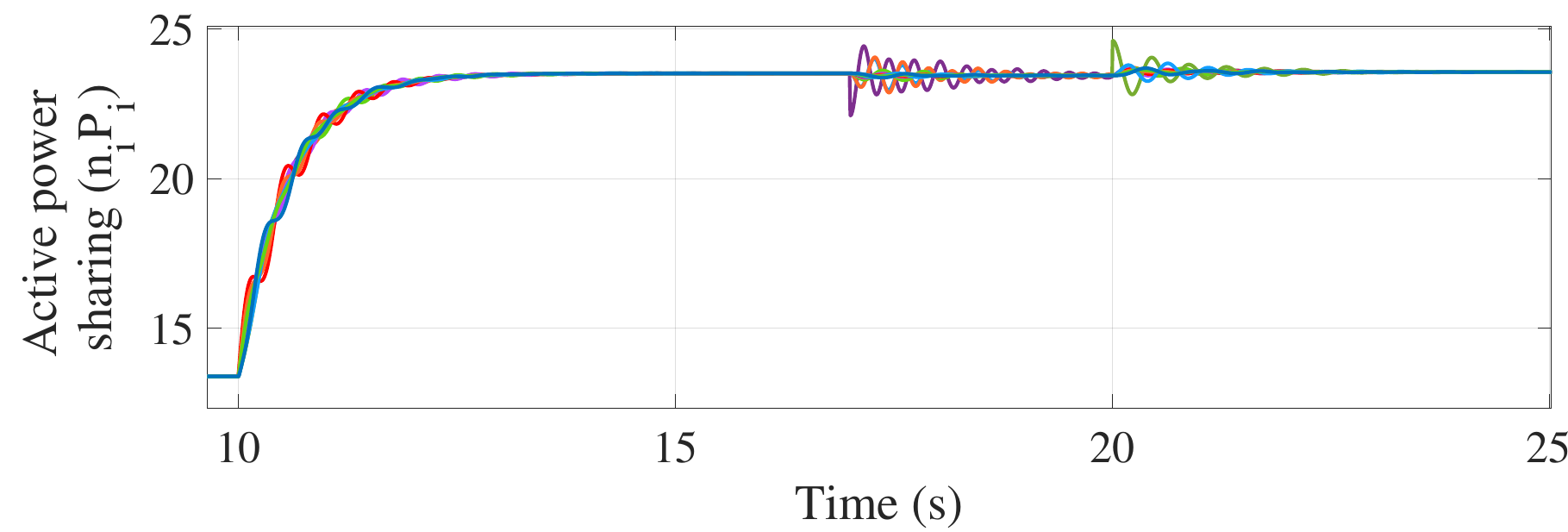}}
\centering \subfigure[]
{\includegraphics[trim=35 240 35 265, clip,width=0.48\textwidth,height=0.165\textwidth]{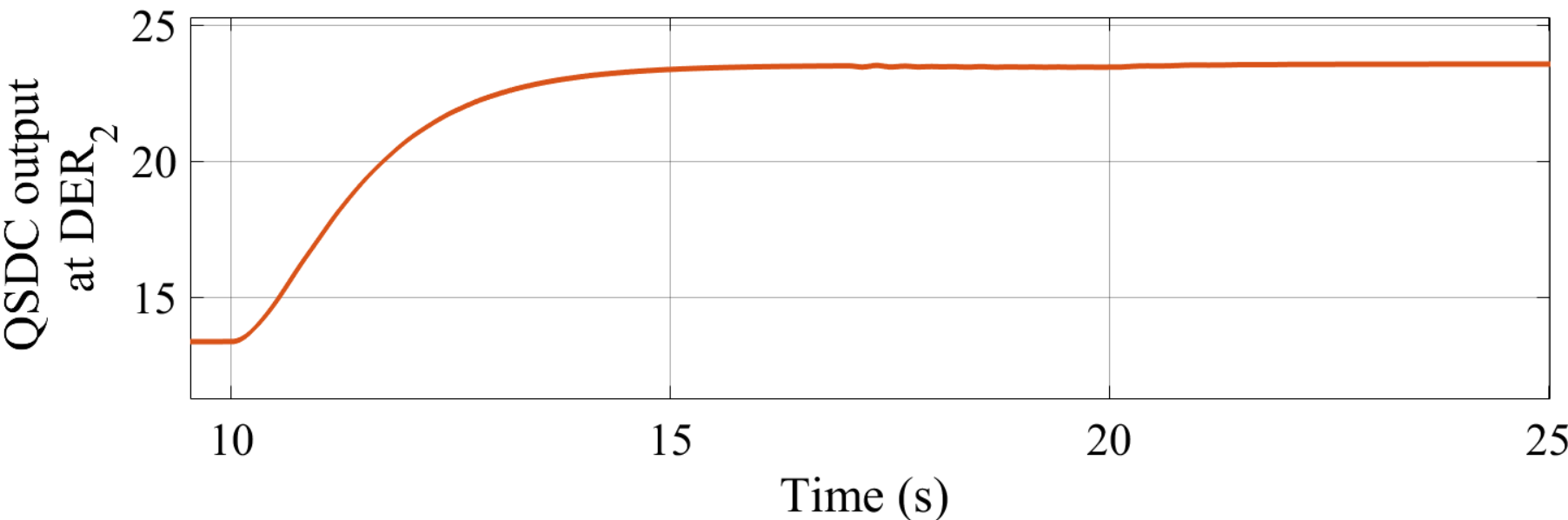}}
\centering \subfigure[]
{\includegraphics[trim=5 0 7 0, clip,width=0.49\textwidth,height=0.195\textwidth]{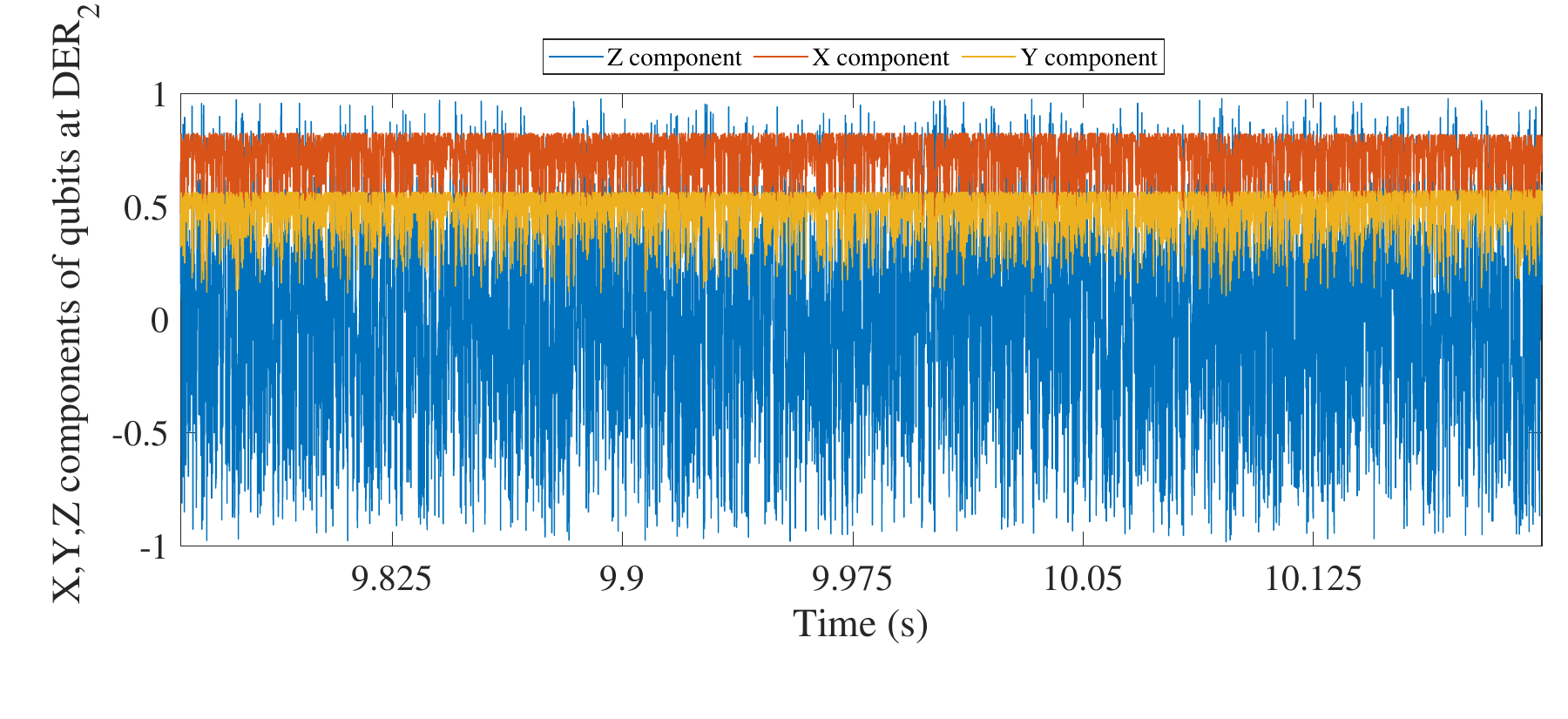}}\vspace{-10pt}
\caption{(a) System's frequency throughout the network after attaching the step load at t = 10s and changing $n_1$ and $n_7$ at t = 17s and t = 20s, respectively. (b) Active power sharing - all $n_iP_i$s converge to the same value. (c) QSDC's output at DER$_2$. (d) x,y and z components of the qubits at DER$_2$.}
\label{s1}
\end{figure*}
\begin{figure}
	\centering
	\includegraphics[trim=37 50 10 35, clip,width=0.5\textwidth, height=0.28\textwidth]{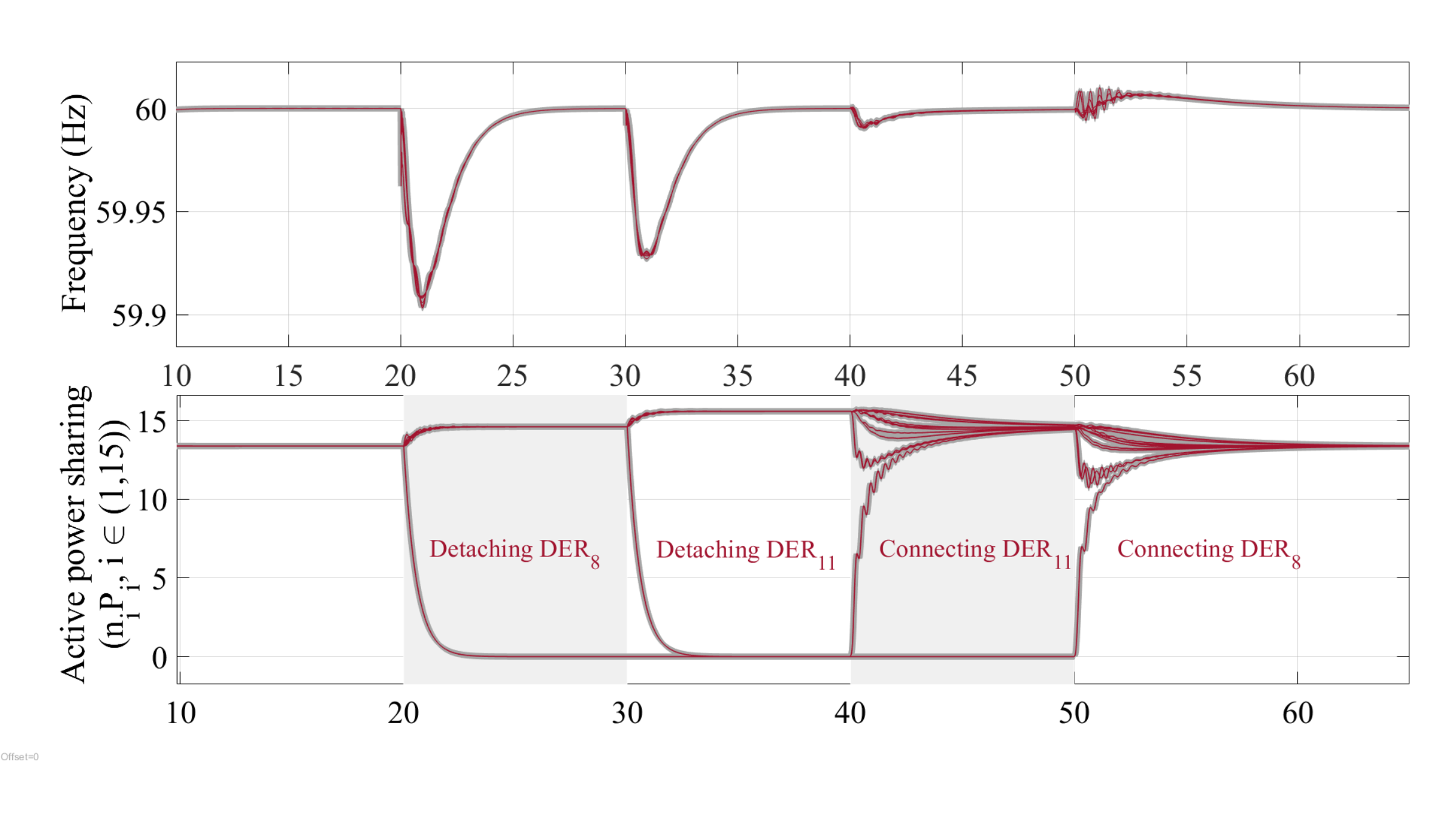}
	\caption{\small Performance of QSDC (gray dash line) and DAPI (red dash line) - Frequency regulation and active power sharing after plug-and-play of DER$_8$ and DER$_{11}$ - At t = 20s and t = 30s the total load is shared among the remaining DERs so their share in power production are increased.}
	\label{fig8}\vspace{-10px}
\end{figure}
\begin{figure}[]
	\centering
	\includegraphics[trim=0 20 0 20, clip,width=0.5\textwidth, height=0.29\textwidth]{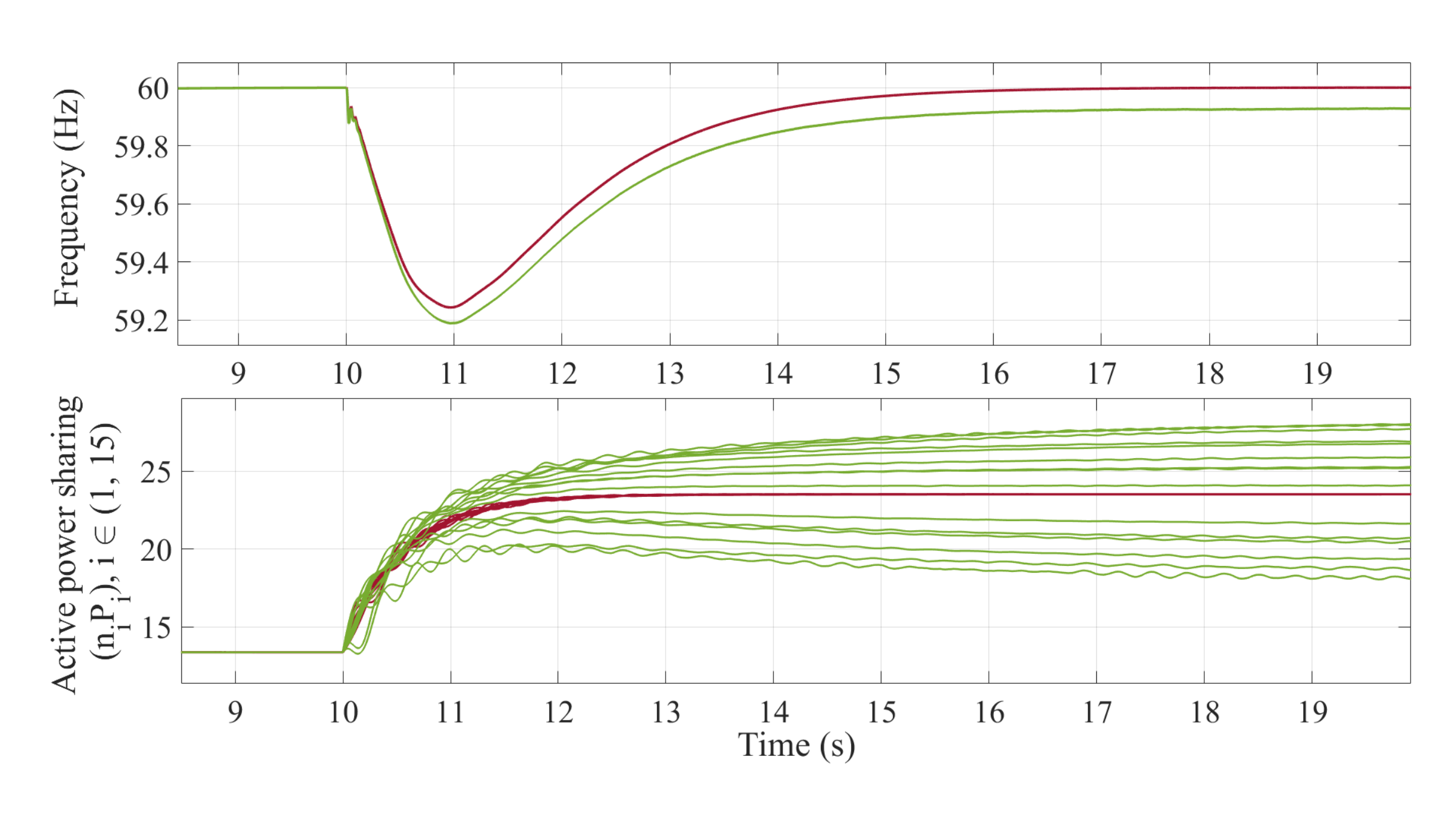}\vspace{-10pt}
	\caption{\small Performance of QSDC (red solid lines) and QDC (green solid lines) under mixed states condition - With QSDC, frequency is regulated to the 60 Hz and all $n_iP_i$ converge to the same value while, with QDC, frequency regulation and power sharing cannot be guaranteed.}
	\label{fig999}%\vspace{-10px}
\end{figure}

\noindent
where $\phi_i/k$ is the secondary control variable and the scaled power sharing signal, $kn_iP_i$, is the pinner. The power sharing signal is scaled to be restricted to $(0, \pi/2)$ through selecting $k$ such that $k < \frac{\pi/2}{\textrm{max}(n_iP_i)}$. 

At the steady state, the microgrid is assumed stable. Since the DERs' frequency must be equal, we have $\omega_i=\omega_j$ and thus $n_iP_{i}-\phi_i/k=n_jP_{j}-\phi_j/k \; \; \; \forall i,j$. As shown before, $\phi_i$ converges to the pinner as $t\rightarrow \infty$. Thus, $n_iP_i=\phi_i/k$ and $n_iP_i=n_jP_j \; \; \; \forall i,j$ and $\omega_i$ converges to $\omega^*$.

\begin{itemize}
    \item \textit{Verification on an AC Networked-Microgrid Case Study}
\end{itemize}

The performance of the developed QSDC is tested on a networked microgrids with five AC microgrids each one has 3 DERs (see Fig.~\ref{figM2}). The nominal voltage and frequency are 380 V and 60 Hz respectively. All other parameters can be found in~\cite{9850415}. For the sake of simulation, three scenarios are examined. First, the controller is evaluated in the face of a step load. To verify the controller's feature of plug-and-play capability, as the second scenario, plug-and-play of DERs is tested. Finally, QSDC is tested when some qubits become mixed.

\subsubsection{Controller Performance}
Studies in this section illustrate the performance of the QSDC under a step load change of 40 kW applied to microgrid 2 at t = 10s and results are depicted in Fig.~\ref{s1}. The exploited communication graph is shown in Fig.~\ref{figM2}. Furthermore, to demonstrate the capability of the controller in handling non-constant droop coefficients, droops at DER$_1$ and DER$_7$ change from $5\times10^{-3}$ to $4.7\times10^{-3}$ and from $2\times10^{-3}$ to $2.1\times10^{-3}$
at t = 17s and t = 20s, respectively. 

As can be seen, frequency regulation is maintained throughout the step load change and Active power is accurately shared among the heterogeneous DERs throughout the entire runtime. Also, despite changing the droop coefficients, the convergence point of all the power sharing signals, $n_iP_i$, doesn’t change as it depends on the power imbalance (power deficiency in this case) at the whole system and not the droop coefficients. 

Furthermore, Fig. \ref{s1}d depicts the randomness of the x, y and z components of the qubits at DER$_2$ during a short period around t = 10s.

\subsubsection{Plug-and-play functionality}
 As the plug-and-play scenario, DER$_8$ and DER$_{11}$ are disconnected at t = 20s and t = 30s, respectively. Then DER$_{11}$ is reconnected at t = 40s followed by the reconnection of DER$_8$ at t = 50s. Fig.~\ref{fig8} shows how after disconnection of DERs, frequency is regulated to the rated 60 Hz. After disconnection, the active power is shared among the connected DERs, and then re-shared again among all the DERs after DER$_8$ and DER$_{11}$ are reconnected.
 
  %For the sake of comparison 
 To compare with a classical scheme, Fig. (\ref{fig8}) also shows the performance of the distributed-averaging PI (DAPI) controller \cite{7112129} under the same plug-and-play scenario. As can be seen, other than providing quantum security on the exchanged information among the participating DERs, the QSDC can guarantee precise power sharing and frequency regulation within a reasonable restoring time compared to the classical solution.

\subsubsection{Robustness against mixed states}

Here, we aim to demonstrate that QSDC can handle the problem of mixed states while, under QDC~\cite{9850415}, active power sharing and frequency regulation cannot be attained. For this case, the scenario $1)$ is repeated for when after t = 10s, qubits at DERs 1, 5 ,9 and 11 become mixed after master equation evolution. Fig. (\ref{fig999}) shows that the performance of the QSDC remains intact (red solid lines) while, with QDC (green solid lines), frequency cannot be regulated to the rated 60 Hz and power sharing is not achieved %causes 
due to excessive exhaustion of some %of the 
participating DERs over time.

\subsection{DC Microgrids}

In DC microgrids, droop control function is mainly utilized to provide decentralized power sharing. It generates the voltage reference $V_i^{{\rm ref}}$ as $V_i^{{\rm ref}}=V^*-m_iI_i$~\cite{6894594}, where $V^*$ is the nominal dc voltage, $m_i$ is the current droop gain and $I_i$ is the output current of DER$_i$. Consider the DC microgrid depicted in Fig. (\ref{fig40}), ignoring the inductance effect of %transmission 
lines, the DC bus voltage $V_b$ can be determined as $V_b=V_i^{{\rm ref}}-R_iI_i$.

It can be shown that, if the current droop gain $m_i$ is set much larger than the line resistance $R_i$, $\frac{I_i}{I_j} \approx \frac{m_i}{m_j}$ and $V_b  \approx V_i^{ref}$ $\forall i,j$. The larger $m_i$ is chosen, the more accurate power sharing can be obtained; however, larger $m_i$ may cause the dc bus voltage $V_b$ to deviate more from the nominal value $V^*$. Therefore, we aim to attain %both 
power sharing and voltage restoration simultaneously, by adding the QSDC. To equip the DC microgrid with the QSDC, the droop function is modified as \vspace{-5pt}
\begin{equation}
    \label{59000}
    \begin{aligned}
    &V_i^{ref}=V^*-m_iI_i+\frac{\phi_i}{c},\\
    &\dot \phi_i=\sin{(cm_iI_i-\phi_i)}+\sum_{j=1}^n a_{i,j}\sin{(\phi_j-\phi_i)},
    \end{aligned}
\end{equation}\vspace{-2pt}

\vspace{-5pt}

Again, we select $c$ such that $c < \frac{\pi/2}{\textrm{max}(m_iI_i)}$. Obviously, the first part in the secondary control dynamic is to drive the dc bus voltage $V_b$ to the nominal value $V^*$ while the second part is to guarantee that $\phi_i=\phi_j$ is satisfied, i.e., the current sharing is achieved which demonstrates that the QSDC is also applicable to distributed voltage control in DC microgrids.

\begin{itemize}
    \item \textit{Verification on a DC Microgrid Case Study}
\end{itemize}

To demonstrate the universality of the QSDC, a 9-DER DC microgrid case study~\cite{9850415} is equipped with QSDC (see Fig.~\ref{fig40}). It is assumed that qubits at DERs 2, 3 and 5 become mixed after t = 20s. First, a step load $R_L=3\Omega$ is applied at t = 20s. Results are depicted in Fig.~(\ref{fig9}a). The exploited communication graph is shown in Fig.~(\ref{fig40}). As can be seen, voltage regulation is guaranteed throughout the step load disturbance. Then, to verify the plug-and-play functionality, at t = 28s DER$_2$ is disconnected and attached again at t = 35s. Fig. (\ref{fig9}a) shows how the current (power) is shared among the participating DERs while voltage regulation is preserved along the time. 

To compare with QDC, the same scenario is repeated for when the microgrid is working under QDC and results are demonstrated in Fig. (\ref{fig9}b). As shown, the problem of mixed states prevents the QDC from guaranteeing voltage regulation and current sharing.

\vspace{-10pt}

\section{Conclusion}

Resilience of the current schemes on classical communication makes microgrids vulnerable to cyber attacks. Inspired by quantum properties of quantum bits, in this paper, we devise a scalable quantum-secure distributed controller that can guarantee 
\begin{figure}
	\centering
	\includegraphics[trim=50 5 5 5,clip, width=0.5\textwidth, height=0.65\textwidth]{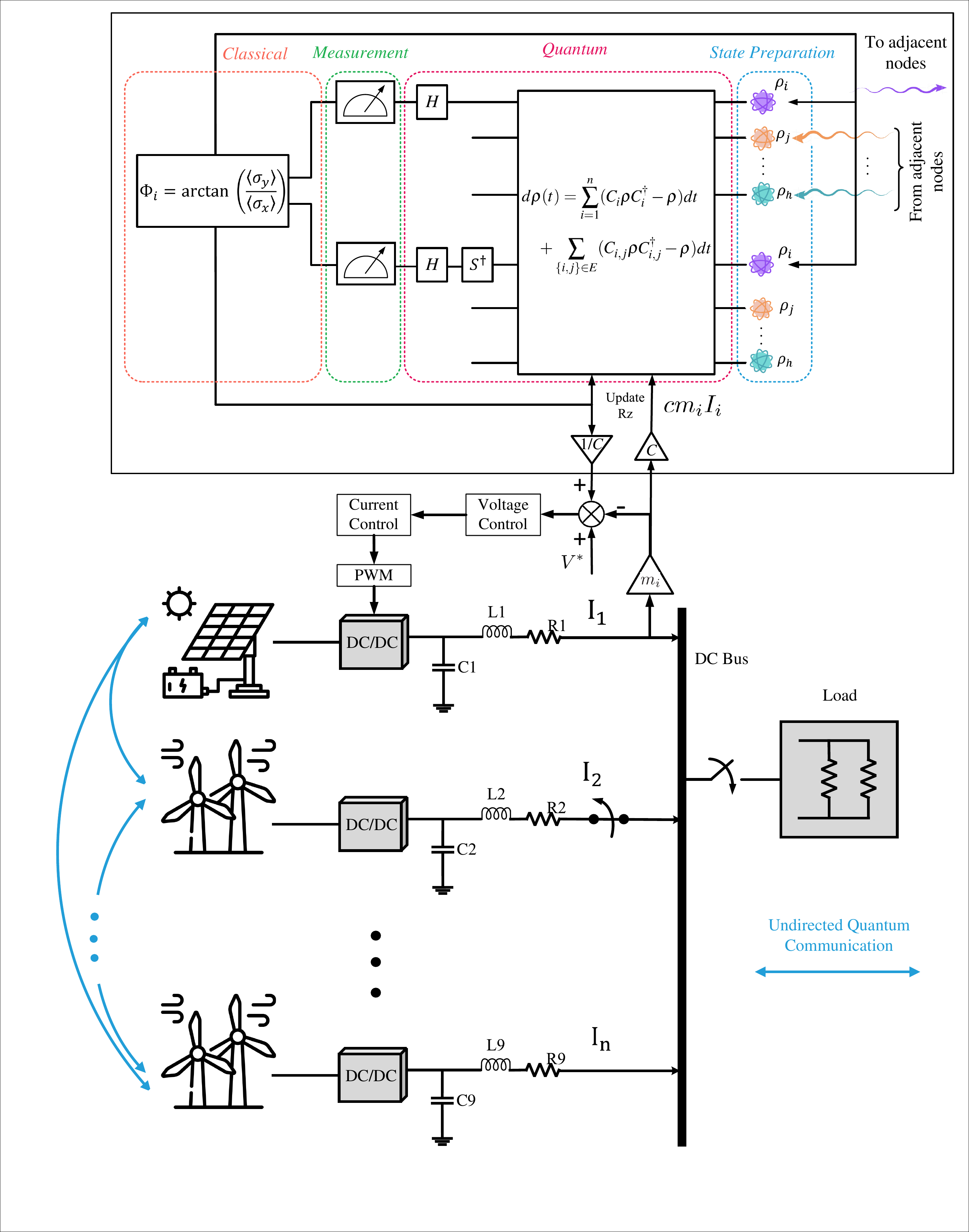} \vspace{-25pt}
	\caption{\small DC microgrid model equipped with the QSDC.}
	\label{fig40}
\end{figure}
\begin{figure}[h]
\centering \subfigure[]
{\includegraphics[trim=10 50 0 5, clip,width=0.49\textwidth, height=0.24\textwidth]{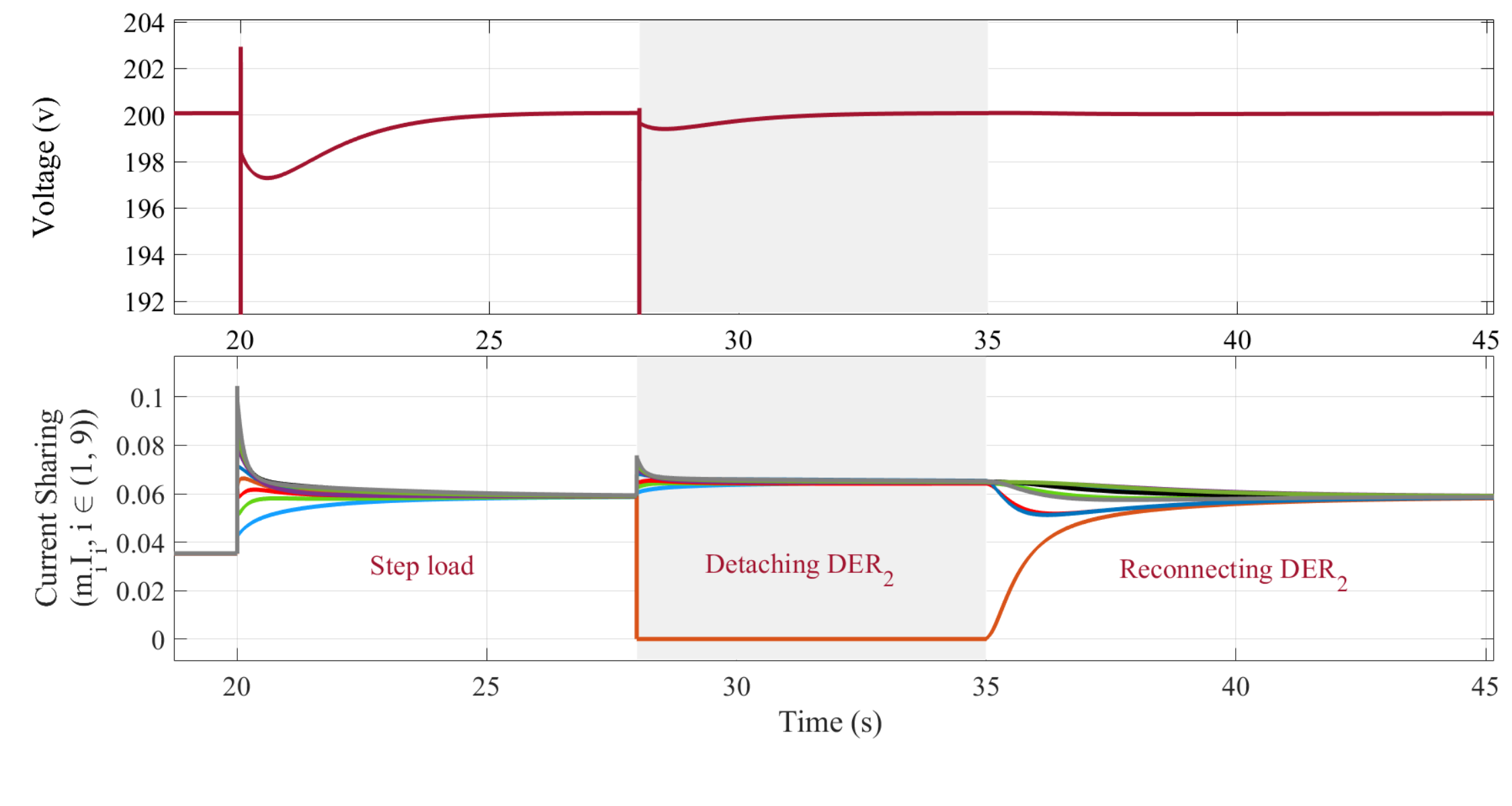}}
\centering \subfigure[]
{\includegraphics[trim=10 30 0 5, clip,width=0.49\textwidth, height=0.24\textwidth]{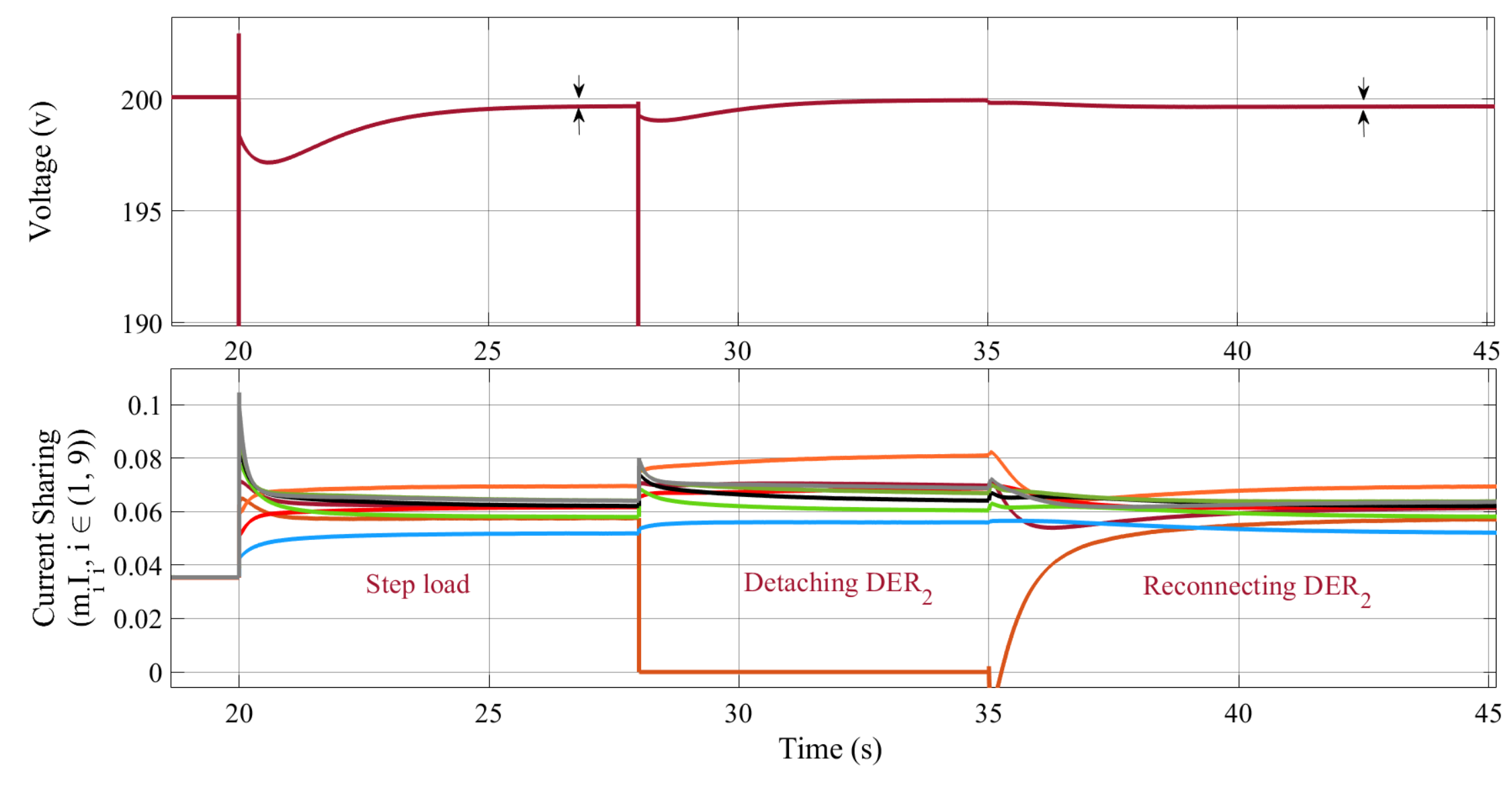}} 
\caption{(a) Voltage regulation and current sharing after a step load disturbance and plug-and play of DER$_2$, under mixed states condition. (b) Performance of QDC under mixed states situation - Voltage regulation and current sharing cannot be guaranteed.}
\label{fig9}
\end{figure}
synchronization and power sharing among participating DERs. We show that in the developed scheme, by allocating random $\theta$ angle to the qubits at the initialization step at each node, the security of the protocol can be unprecedentedly enhanced since, the measurement outcomes of the exchanged qubits, obtained by a third-party agent, would be some random values which do not reveal meaningful information. The new quantum distributed control scheme thus assures provably cyber-resilient microgrids.

\vspace{-10pt}

\appendices

\section{Preliminaries}

\subsubsection{Graph Theory}

%Some basic concepts from graph theory \cite{Godsil} are provided here. 
A simple graph $G = (V, E)$ consists of a set of $n$ nodes, $V=\{ v_1, v_2, \ldots, v_n\}$, and a set of edges, $E \subset V \times V$. An edge $(v_i,v_j) \subset E$ represents that nodes $v_i$ and $v_j$ can exchange information with each other. 
%The graph $\mathcal{G}$ is said to be undirected if for all $v_i,v_j\in E_{\mathcal{G}}$, $(v_j,v_i) \subset E(\mathcal{G})$ whenever $(v_i,v_j) \subset E(\mathcal{G})$. 
A sequence of non-repeated edges $(v_i, v_{p_1}), (v_{p_1}, v_{p_2}), \ldots, (v_{p_{m-1}}, v_{p_m}), (v_{p_m},v_j)$ is called a path between nodes $v_i$ and $v_j$. If there exists a path between any two different nodes $v_i,v_j \in V$, $G$ is called connected.
A node $v_j$ is called a neighbor of node $v_i$ if $(v_j, v_i) \subset E$. The set of neighbors of node $v_i$ is denoted as $N_i=\{ v_j \in V\; | \;(v_j, v_i) \subset E\}$. The adjacency matrix of graph $G$, denoted as $A$, is an $n\times n$ matrix whose entries $a_{i,j}=1$ if $v_j \in N_i$ and $a_{i,j}=0$ otherwise. The degree matrix $D$ of graph $G$, denoted as $D$, is defined as an $n\times n$ diagonal matrix whose $i$th diagonal entry equals the degree of node $v_i$, i.e., $\sum_{v_j \in N_i} a_{i,j}$. The Laplacian matrix of graph $G$, denoted as $L$, is defined as $D-A$. Note that $A,D,L$ are all symmetric. The node-edge incidence matrix $B \in R^{V \times E}$ is defined component-wise as $B_{i,j}=1$ if edge $j$ enters node $i$, $B_{i,j}=-1$ if edge $j$ leaves node $i$, and $B_{i,j}=0$ otherwise. For $x \in R^{V}$, $B^Tx \in R^E$ is the vector with components $x_i-x_j$, with $\{i,j\} \in E$. If $\textrm{diag}(\{a_{i,j}\}_{\{i,j\}\in E})$ is the diagonal matrix of edge weights, then the Laplacian matrix is given by $L=B\textrm{diag}(\{a_{i,j}\}_{\{i,j\}\in E})B^T$.

\subsubsection{Quantum Systems and Notations}

%Throughout this paper, 
The (adjoint) $\dagger$ symbol indicates the transpose-conjugate in matrix representation, and the tensor product $\otimes$ is associated to the Kronecker product. 

The mathematical description of a single quantum system starts by considering a complex Hilbert space $\mathcal{H}$. We utilize Dirac's notation, where $\ket{\psi}$ denotes an element of $\mathcal{H}$, called a ket which is represented by a column vector, while $\bra{\psi}=\ket{\psi}^{\dagger}$ is used for its dual, a bra, represented by a row vector, and $\bra{\psi}\ket{\varphi}$ for the associated inner product. We denote the set of linear operators on $\mathcal{H}$ by $\mathfrak{B}(\mathcal{H})$. The adjoint operator $X^{\dagger} \in \mathfrak{B}(\mathcal{H})$ of an operator $X \in \mathfrak{B}(\mathcal{H})$ is the unique operator that satisfies $(X \ket{\psi})^{\dagger}\ket{\chi}=\bra{\psi}(X^{\dagger}\ket{\chi})$ for all $\ket{\psi}$, $\ket{\chi} \in \mathcal{H}$. The natural inner product in $\mathfrak{B}(\mathcal{H})$ is the Hilbert-Schmidt product $\langle X, Y\rangle=\mathrm{tr}(X^{\dagger}Y)$, where $tr$ is the usual trace functional which is canonically defined in a finite dimensional setting. We denote by $I$ the identity operator. $[A,B]=AB-BA$ is the commutator and $\{ A, B\} = AB + BA$ is the anticommutator of $A$ and $B$.

Qubit, defined as the quantum state of a two-state quantum system, is the smallest unit of information, and it is analogous to classical bit. State of a qubit, represented by $\ket{\psi}=\alpha\ket{0}+\beta\ket{1}$, is superposition of the two orthogonal basis states $\ket{0}\sim[1,\; 0]^T$ and $\ket{1}\sim[0,\; 1]^T$. $\alpha$ and $\beta$ are complex numbers in general, where $\left|\alpha\right|^2 + \left|\beta\right|^2=1$. We denote $\ket{q_1}\otimes...\otimes\ket{q_n} \in \mathcal{H}^{\otimes n}$ as $\ket{q_1...q_n}$ where $\otimes$ represents the tensor product. Since $\left|\alpha\right|^2 + \left|\beta\right|^2=1$, we can represent the state of a qubit by $\ket{\psi}=\cos{\frac{\theta}{2}}\ket{0}+e^{\imath \phi}\sin{\frac{\theta}{2}}\ket{1}$ where $\theta$ and $\phi$ are real numbers defining a point on a unit three-dimensional sphere called Bloch sphere.
 
In the case of mixed state, the state of a quantum system is represented by a $\textit{density operator}$ $\rho$, that is any self-adjoint positive semi-definite operator with trace one, and $\rho = \ket{\psi}\bra{\psi}$ with $\ket{\psi} \in \mathcal{H}$ and $\bra{\psi}\ket{\psi}=1$ are called pure states. For further information on qubits see \cite{nielsen_chuang_2010, preskill}.

\vspace{-3pt}

\section{Proof of Convergence for eq. \eqref{260b}}

For synchronization of the system, it is critical that the pinner for all of the oscillators be the same, i.e., $\phi_{t,i}=\phi^*$, otherwise, synchronization cannot be achieved in general. To study whether quantum node $i$ is synchronized to the pinner, it is convenient to study the phase deviation of quantum node $i$ from the pinner. We introduce the following change of variables,\vspace{-2pt}
\begin{equation}
    \label{eq37}
    \begin{aligned}
    \phi_i=\phi^*+\zeta_i,
    \end{aligned}
\end{equation}
where $\zeta_i$ denotes the phase deviation of the $i$th oscillator from the pinner $\phi^*$. Substituting (\ref{eq37}) into (\ref{eq35}), we have \vspace{-2pt}
\begin{equation}
    \label{eq38}
    \begin{aligned}
    \dot \zeta_i = \sum_{j=1}^n a^{'}_{i,j}\sin{(\zeta_j-\zeta_i)}-\sin{(\zeta_i)}.
    \end{aligned}
\end{equation}
By studying the properties of (\ref{eq38}), we can obtain the condition for synchronization. If all $\zeta_i$'s converge to 0, then we have $\phi_i=\phi^*$ as $t\rightarrow \infty$, indicating that all nodes are synchronized to the pinner. Let $B=[B_{i,j}]_{n\times m}$ be the incidence matrix~\cite{Godsil} of the communication graph $G$ with
$m$ being the number of edges. Then, (\ref{eq38}) can be recast in a state form:
\begin{equation}
    \label{eq40}
    \begin{aligned}
    \dot \zeta = -\sin{\zeta}-BW\sin(B^T\zeta),
    \end{aligned}
\end{equation}
where $W=\textrm{diag}(\{a^{'}_{i,j}\}_{\{i,j\}\in E})$ is the diagonal matrix of edge weights and $\sin (\cdot)$ takes entry-wise operation for a vector. 

To proceed, set $\varepsilon=\max\limits_{1\leq i\leq n}|\zeta_i|$. When $\varepsilon<(\pi/2)$, if $\zeta_i=\varepsilon$, we have $-\pi<-2\varepsilon\leq \zeta_j-\zeta_i\leq 0$ for $1\leq j \leq n$. Hence, in (\ref{eq38}), $\sin{(\zeta_j-\zeta_i)}\leq0$ and $\sin{\zeta_i}>0$ hold, and hence $\dot \zeta_i<0$ hold. Therefore, the vector field is pointing inward in the set, and no trajectory can escape to values larger than $\varepsilon$. Similarly, it can be obtained that, when $\zeta_i=-\varepsilon$, $\dot \zeta_i>0$ holds. Thus no trajectory can escape to values smaller than $-\varepsilon$. Therefore, $\zeta \in [-\varepsilon, \varepsilon]\times \cdots \times [-\varepsilon, \varepsilon]=[-\varepsilon, \varepsilon]^n $ is positively invariant when $\varepsilon<\pi/2$, where $\times$ denotes Cartesian product.
%Now we proceed to prove the synchronization.
Define a Lyapunov function $V = (1/2) \zeta^T\zeta$, which equals zero only if all $\zeta_i$ are zero, meaning the synchronization of all nodes to the pinner. Differentiating $V$ along the trajectories of (\ref{eq40}) yields
\begin{equation}\label{eq41}
    \dot V %&= \zeta^T\dot \zeta=-\zeta^T\left (\sin{\zeta+BW\sin{B^T\zeta}} \right )\\
    =-\zeta^TS_1\zeta-\zeta^TBWS_2B^T\zeta,
\end{equation}
where $S_1\in R^{n\times n}$ and $S_2\in R^{m\times m}$ are given by
\begin{equation}
    \label{eq42}
    \begin{aligned}
    &S_1=\textrm{diag} \left\{ \textrm{sinc}(\zeta_1), \cdots, \textrm{sinc}(\zeta_n) \right\},\\
    &S_2=\textrm{diag} \left\{ \textrm{sinc}(B^T\zeta)_1, \cdots, \textrm{sinc}(B^T\zeta)_m \right\},
    \end{aligned}
\end{equation}
where $\textrm{sinc(x)}\equiv\textrm{sin(x)}/\textrm{x}$ and $(B^T\zeta)_i$ denotes the $i$th element of $m$-dimensional vector $B^T\zeta$. 
%From dynamic systems theory, if $S_1+BWS_2B^T$ in (\ref{eq41}) is positive definite when $\zeta \neq 0$, then $\dot V$ is negative when $\zeta \neq 0$ and $V$ will decay to 0, meaning that $\zeta$ will decay to 0 and all nodes are synchronized to the pinner.
When all $\zeta_i$ are within $[-\varepsilon, \varepsilon]$ with $0\leq \varepsilon<(\pi/2)$, $(B^T\zeta)_i$ is in the form of $\zeta_k-\zeta_l \; (1\leq k, l\leq n)$, and hence is restricted to $(-\pi, \pi)$. Given that in $(-\pi, \pi)$, $\textrm{sinc(x)}>0$ holds, it follows that 
$S_1\geq \sigma_1I$ and $S_2\geq \sigma_2I$ with $\sigma_1= \textrm{sinc}(\varepsilon)$ and $\sigma_2=  \textrm{sinc}(2\varepsilon)$.
Thus, $S_1+BWS_2B^T\geq \sigma_1I+\sigma_2BWB^T$, which in combination with (\ref{eq41}) yields 
$\dot V \leq -\zeta^T(\sigma_1I+\sigma_2BWB^T)\zeta$.
Note that $BWB^T$ is the Laplacian matrix of the underlying graph $G$, which is always positive semidefinite. 
Since $\sigma_1$ and $\sigma_2$ are positive,  

$\sigma_1I+\sigma_2BWB^T$ must be positive definite. It follows that, when $0\leq \varepsilon<(\pi/2)$, we have $\dot V\leq -2\mu V$, where 
$\mu =\lambda_{{\rm min}}(\sigma_1I+\sigma_2BWB^T)>0$,
which implies that all the nodes will synchronize to the pinner exponentially fast at a rate no less than $\mu$, which is dependent on the network connectivity.

\bibliographystyle{IEEEtran}

% \bibliography{structure.bib}

\end{document}